%% file: main.tex
\documentclass[letterpaper]{article}

\usepackage[margin=1.5in]{geometry}
\usepackage{graphicx}
\usepackage{latexsym}
\usepackage{amsmath}
\usepackage{amsfonts,bm}
\usepackage{amsthm}
\usepackage{amssymb}
\usepackage{verbatim}
\usepackage{xspace}

\usepackage[linesnumbered,boxed]{algorithm2e}
\usepackage{epsfig}
\usepackage[usenames]{color}

\definecolor{darkgreen}{rgb}{0,0.5,0}
\definecolor{darkblue}{rgb}{0,0,0.8}
\usepackage{hyperref}
\hypersetup{
    unicode=false,          
    colorlinks=true,        
    linkcolor=darkblue,          
    citecolor=darkgreen,        
    filecolor=magenta,      
    urlcolor=cyan           
}

\usepackage[capitalize]{cleveref}

\newtheorem{theorem}{Theorem}
\newtheorem{lemma}[theorem]{Lemma}
\newtheorem{corollary}[theorem]{Corollary}

\newtheorem{example}{Example}

\newtheorem{definition}{Definition}

\newcommand{\func}[2]{{#1}\!\left(#2\right)}

\newcommand{\B}[3]{\mathcal{B}_{#1,#2,#3}}
\newcommand{\Bk}[4]{\B{#1}{#2}{#3}^{#4}}
\newcommand{\Bup}[3]{\Bk{#1}{#2}{#3}{\scriptscriptstyle{\uparrow}}}
\newcommand{\Bdown}[3]{\Bk{#1}{#2}{#3}{\scriptscriptstyle{\downarrow}}}

\newcommand{\M}[1]{M_{#1}}

\newcommand{\requal}[1]{\;\overset{\scriptscriptstyle{#1}}{\scriptscriptstyle{=}}\;}

\newcommand{\vektor}[1]{\ensuremath{\underline{#1}}}
\newcommand{\lp}{\ensuremath{\mathrm{LP}}}

\newcommand{\ijmin}{i_{j_\mathrm{min}}}

\newcommand{\LS}{\ensuremath{\mathcal{LS}}}

\newcommand{\set}[1]{\left\{ #1 \right\}}
\newcommand{\eps}{\ensuremath{\varepsilon}}

\DeclareMathOperator{\Oh}{O}
\DeclareMathOperator{\E}{E}
\DeclareMathOperator{\prob}{Pr}

\title{Local Computation: Lower and Upper Bounds\footnote{This paper
    is based in part on work that has appeared in the following two
    preliminary versions: \emph{What Cannot Be Computed Locally}, In
    \emph{Proceedings of the 23rd ACM Symposium on the Principles of
      Distributed Computing (PODC)}, St. John's, Canada,
    2004~\cite{kuhn-podc04} and \emph{The Price of Being
      Near-Sighted}, In \emph{Proceedings of the 17th ACM-SIAM
      Symposium on Discrete Algorithms (SODA)}, Miami, Florida,
    2006~\cite{kuhn-soda06}. We are grateful to Bar-Yehuda,
    Censor-Hillel, and Schwartzman \cite{baryehuda16} for pointing out
    an error in an earlier draft \cite{previousversion} of this
    paper.}}

\author {
  Fabian Kuhn$^1$,  Thomas Moscibroda$^2$, Roger Wattenhofer$^3$\\[2mm]
   $^1$kuhn@cs.uni-freiburg.de, University of Freiburg, Germany\\
   $^2$moscitho@microsoft.com, Microsoft Research, Beijing, China\\
  $^3$wattenhofer@ethz.ch, ETH Zurich, Switzerland}

\date{}

\begin{document}

\maketitle

\begin{abstract}
  The question of what can be computed, and how efficiently, are at
  the core of computer science. Not surprisingly, in distributed
  systems and networking research, an equally fundamental question is
  what can be computed in a \emph{distributed} fashion. More
  precisely, if nodes of a network must base their decision on
  information in their local neighborhood only, how well can they
  compute or approximate a global (optimization) problem? In this
  paper we give the first poly-logarithmic lower bound on such local
  computation for (optimization) problems including minimum vertex
  cover, minimum (connected) dominating set, maximum matching, maximal
  independent set, and maximal matching. In addition we present a new
  distributed algorithm for solving general covering and packing
  linear programs. For some problems this algorithm is tight with the
  lower bounds, for others it is a distributed approximation scheme.
  Together, our lower and upper bounds establish the local
  computability and approximability of a large class of problems,
  characterizing how much local information is required to solve these
  tasks.
\end{abstract}

\renewcommand{\include}{\input}

\include{introduction}

\include{model}

\include{problems}

\include{contribution}

\include{relatedwork}

\include{lowerbound}

\include{reductions}

\include{upperbound}


\include{conclusions}

\bibliographystyle{abbrv}
\bibliography{references}

\end{document}

%% file: introduction.tex
\section{Introduction}\label{sec:intro}

Many of the most fascinating systems in the world
are large and complex networks, such as the human society, the
Internet, or the brain. Such systems have in common that they are
composed of a multiplicity of individual entities,
so-called \emph{nodes}; human beings in society, hosts in the
Internet, or neurons in the brain.
Each individual node can directly communicate only to a small number
of neighboring nodes. For instance, most human communication is
between acquaintances or within the family, and neurons are directly
linked with merely a relatively small number of other neurons. On
the other hand, in spite of each node being inherently
``near-sighted,'' i.e., restricted to \emph{local} communication,
the entirety of the system is supposed to work towards some kind of
\emph{global} goal, solution, or equilibrium.

In this work we investigate the possibilities and limitations of
\emph{local computation}, i.e., to what degree local information is
sufficient to solve global tasks. Many tasks can be solved entirely locally, for
instance, how many friends of friends one has. Clearly, only local communication is required to answer this question. Many other tasks are
inherently global, for instance, counting the total number of nodes or determining the diameter of the system. To solve such global
problems, some information must traverse across the entire network.

Are there natural tasks that are in the middle
of these two extremes, tasks that are neither completely local nor
inherently global? In this paper we answer this question
affirmatively. Assume for example that the nodes want to organize
themselves, some nodes should be masters, the others will be slaves.
The rules are that no two masters shall be direct neighbors, but
every slave must have at least one master as direct neighbor. In
graph theory, this problem is known as the \emph{maximal independent
set} (MIS) problem. At first, this problem seems local since the
rules are completely local. Consequently one might hope for a solution where
each node can communicate with its neighbors a few times, and
together they can decide who will become master and who will become
slave. However, as we show in this paper, this intuition is
misleading. Even though the problem can be defined in a purely local way, it cannot be solved
using local information only! No matter how the system tackles the
problem, no matter what protocol or algorithm the nodes use,
non-local information is vital to solve the task. On the other hand,
the problem is also not global: Mid-range information is enough to
solve the problem. As such the MIS problem establishes an example
that is neither local nor global, but in-between these extremes.
As it turns out to be polylogarithmic in the number of nodes, we call it \emph{polylog-local}.
Using \emph{locality-preserving reductions} we are able to show that
there exists a whole class of polylog-local problems.

We show that this class of polylog-local problems also includes
approximation variants of various combinatorial optimization problems, such as minimum vertex cover,
minimum dominating set, or maximum matching. In such problems, each
node must base its decision (for example whether or not to join the
dominating set) only on information about its local neighborhood,
and yet, the goal is to collectively achieve a good approximation to
the globally optimal solution. Studying such \emph{local
approximation algorithms} is particularly interesting because it
sheds light on the trade-off between the amount of available local
information and the resulting global optimality. Specifically, it
characterizes the amount of information needed in distributed
decision making: what can be done with the information that is
available within some fixed-size neighborhood of a node. Positive
and negative results for local algorithms can thus be interpreted as
information-theoretic upper and lower bounds; they give insight into
the value of information.

We believe that studying the fundamental possibilities and
limitations of local computation is of interest to theoreticians in
approximation theory, distributed computing, and graph theory.
Furthermore, our results may be of interest for a wide range of
scientific areas, for instance dynamic systems that change over
time. Our theory shows that small changes in a dynamic system may
cause an intermediate (or polylog-local) ``butterfly effect,'' and it
gives non-trivial bounds for self-healing or self-organizing
systems, such as self-assembling robots. It also establishes bounds
for further application areas, initially in engineering and
computing, possibly extending to other areas studying large-scale
systems, e.g., social science, finance, neural networks, or ant
colonies.

%% file: model.tex
\subsection{Model and Notation}\label{sec:model}

\textbf{Local Computations: }We consider a distributed system in
which distributed decision makers at the nodes of a graph must base
their computations and decisions on the knowledge about their local
neighborhoods in the graph. Formally, we are given a graph
$G=(V,E)$, $|V|=n$, and a parameter $k$ ($k$ might depend on $n$ or
some other property of $G$). At each node $v\in V$ there is an
independent agent (for simplicity, we identify the agent at node $v$
with $v$ as well). Every node $v\in V$ has a unique identifier
$id(v)$\footnote{All our
  results hold for any possible ID space including the standard case
  where IDs are the numbers $1,\dots,n$.} and possibly some additional
input. We assume that each node $v\in V$ can learn the complete
neighborhood $\Gamma_k(v)$ up to distance $k$ in $G$ (see below for
a formal definition of $\Gamma_k(v)$). Based on this information,
all nodes need to make independent computations and need to
individually decide on their outputs without communicating with each
other. Hence, the output of each node $v\in V$ can be computed as a
function of it's $k$-neighborhood $\Gamma_k(v)$.

\textbf{Synchronous Message Passing Model: }The described
graph-theoretic local computation model is equivalent to the classic
\emph{message passing} model of distributed computing. In this
model, the distributed system is modeled as a point-to-point
communication network, described by an undirected graph $G=(V,E)$,
in which each vertex $v\in V$ represents a node (host, device,
processor, \ldots) of the network, and an edge $(u,v)\in E$ is a
bidirectional communication channel that connects the two nodes.
Initially, nodes have no knowledge about the network graph; they
only know their own identifier and potential additional inputs. All
nodes wake up simultaneously and computation proceeds in synchronous
\emph{rounds}. In each round, every node can send one, arbitrarily
long message to each of its neighbors. Since we consider
point-to-point networks, a node may send different messages to
different neighbors in the same round. Additionally, every node is
allowed to perform local computations based on information obtained
in messages of previous rounds. Communication is reliable, i.e.,
every message that is sent during a communication round is correctly
received by the end of the round.  An algorithm's \emph{time
complexity} is defined as the number of communication rounds until
all nodes terminate.\footnote{Notice that this synchronous message passing model captures many practical systems, including for example, Google's Pregel system, a practically implemented computational model suitable for computing problems in large graphs~\cite{pregel}. }

The above is a standard model of distributed computing and is
generally known as the LOCAL model~\cite{pelegbuch,linial92}. It is
the strongest possible model when studying the impact of
locally-restricted knowledge on computability, because it focuses
entirely on the locality of distributed problems and abstracts away
other issues arising in the design of distributed algorithms (e.g.,
need for small messages, fast local computations, congestion,
asynchrony, packet loss, etc.). It is thus the most fundamental
model for proving lower bounds on local computation~\cite{linial92};
because any lower bound is a true consequence of locality
restrictions.

\textbf{Equivalence of Time Complexity and Neighborhood-Information:
}There is a one-to-one correspondence between the \emph{time
complexity of distributed algorithms} in the LOCAL model and the
graph theoretic notion of \emph{neighborhood-information}. In
particular, a distributed algorithm with time-complexity $k$ (i.e.,
in which each node performs $k$ communication rounds) is equivalent
to a scenario in which distributed decision makers at the nodes of a
graph must base their decision on (complete) knowledge about their
$k$-hop neighborhood $\Gamma_k(v)$ only. This is true because with
unlimited sized messages, every node $v\in V$ can easily collect all
IDs and interconnections of all nodes in its $k$-hop neighborhood in
$k$ communication rounds.
On the other hand, a node $v$ clearly cannot obtain any information
from a node at distance $k+1$ or further away, because this
information would require more than $k$ rounds to reach $v$. Thus,
the LOCAL model relates distributed computation to the
\emph{algorithmic theory of the value of information} as studied for
example in \cite{papadimitriou-stoc93}: the question of \emph{how
much
  local knowledge} is required for distributed decision makers to
solve a global task or approximate a global goal is equivalent to
the question of \emph{how many communication rounds} are required by a
distributed algorithm to solve the task.

\textbf{Notation: } For nodes $u,v\in V$ and a graph $G=(V,E)$,
we denote the shortest-path distance between $u$ and $v$ by
$d_G(u,v)$. Let $\Gamma_k(v)$ be the $k$-hop neighborhood of a
node $v\in V$. Formally, we define $\Gamma_k(v)
:= \set{u\in V:d_G(u,v)\leq k}$. We also use the shortcut
$\Gamma_v:=\Gamma_1(v)$, that is, $\Gamma_v$ is the (inclusive) neighborhood of
$v$. In a local computation with $k$-hop neighborhood information (or
equivalently, in any distributed algorithm with time complexity $k$),
each node has a \emph{partial view} of the graph and must base its
algorithm's outcome solely on information obtained in
$\Gamma_k(v)$. Formally, let $\mathcal{T}_{v,k}$ be the topology seen
by $v$ after $k$ rounds in a distributed algorithm, i.e.,
$\mathcal{T}_{v,k}$ is the graph induced by the $k$-neighborhood of
$v$ where edges between nodes at exactly distance $k$ are excluded.
The \emph{labeling} (i.e., the assignment of identifiers to nodes) of
$\mathcal{T}_{v,k}$ is denoted by
$\mathcal{L}(\mathcal{T}_{v,k})$. The \emph{view} of a node $v$ is the
pair
$\mathcal{V}_{v,k}:=(\mathcal{T}_{v,k},\mathcal{L}(\mathcal{T}_{v,k})
)$. Any deterministic distributed algorithm can be regarded as a
function mapping $(\mathcal{T}_{v,k},\mathcal{L}(\mathcal{T}_{v,k}))$
to the possible outputs. For randomized algorithms, the outcome of $v$
is also dependent on the randomness computed by the nodes in
$\mathcal{T}_{v,k}$.


%% file: problems.tex
\subsection{Problem Definitions}
\label{sec:problems}


In this paper, we study several standard combinatorial optimization
problems (and their natural relaxations) that intuitively appear to
be local, yet turn out to be neither completely local nor global.
Specifically, we consider the following standard optimization
problems in graphs:
\begin{itemize}
\item \textbf{Minimum Vertex Cover (MVC):} Given a graph $G=(V,E)$,
  find a minimum vertex subset $S\subseteq V$, such that for each edge in
  $E$, at least one of its endpoints is in $S$.
\item \textbf{Minimum Dominating Set (MDS):} Given a graph $G=(V,E)$,
  find a minimum vertex subset $S\subseteq V$, such that for each
  node $v\in V$, either $v\in S$ or at least one neighbor of $v$ must
  be in $S$.
\item \textbf{Minimum Connected Dominsting Set (MCDS):} Given a graph
  $G=(V,E)$, find a minimum dominating set $S\subseteq V$, such that
  the graph $G[S]$ induced by $S$ is connected.
\item \textbf{Maximum Matching (MaxM):} Given a graph $G=(V,E)$, find a
  maximum edge subset $T\subseteq E$, such that no two edges in $T$
  are adjacent.
\end{itemize}
In all these cases, we consider the respective problem on the network
graph, i.e., on the graph representing the network. In addition to the
above mentioned problems, we study their natural linear programming
relaxations as well as a slightly more general class of linear programs
(LP) in a distributed context. Consider an LP and its corresponding
dual LP in the following canonical forms:

\hspace*{\fill}\parbox{45mm}{
  \begin{align*}
    \min && \vektor{c}^\mathrm{T}\vektor{x} & \\
    \text{s.\ t.} && A\cdot\vektor{x} & \ge \vektor{b} \\
    && \vektor{x} & \geq \vektor{0}.
  \end{align*}}\hfill
\parbox{1cm}{\hfill(P)}\ \
\hspace*{\fill}\parbox{45mm}{
  \begin{align*}
    \min && \vektor{b}^\mathrm{T}\vektor{y} & \\
    \text{s.\ t.} && A^{\mathrm{T}}\cdot\vektor{y} & \le \vektor{c} \\
    && \vektor{y} & \geq \vektor{0}.
  \end{align*}}\hfill
\parbox{1cm}{\hfill(D)}

We call an LP in form (P) to be in primal canonical form (or just in
canonical form) and an LP in form (D) to be in dual canonical form. If
all the coefficients of $\vektor{b}$, $\vektor{c}$, and $A$ are
non-negative, primal and dual LPs in canonical forms are called
\emph{covering} and \emph{packing} LPs, respectively. The
relaxations of vertex cover and dominating set are covering LPs,
whereas the relaxation of matching is a packing LP.

While there is an obvious way to interpret graph problems such as
vertex cover, dominating set, or matching as a distributed problem,
general LPs have no immediate distributed meaning.  We use a natural
mapping of an LP to a network graph, which was introduced
in~\cite{papadimitriou-stoc93} and applied in~\cite{bartal97}. For
each primal variable $x_i$ and for each dual variable $y_j$, there
are nodes $v_i^p$ and $v_j^d$, respectively. We denote the set of
primal variables by $V_p$ and the set of dual variables by $V_d$.
The network graph $G_{\lp}=(V_p\dot{\cup} V_d,E)$ is a bipartite
graph with the edge set
\[
E\ :=\ \left\{(v_i^p,v_j^d)\in V_p\times V_d\,\big|\ a_{ji}\not=0\right\},
\]
where $a_{ji}$ is the entry of row $j$ and column $i$ of $A$. We
define $n_p:=|V_p|$ and $n_d:=|V_d|$, that is, $A$ is a $(n_d\times
n_p)$-matrix. Further, the maximum primal and dual degrees are denoted
by $\Delta_p$ and $\Delta_d$, respectively. In most real-world
examples of distributed LPs and their corresponding combinatorial
optimization problems, the network graph is closely related to the
graph $G_{\lp}$ such that any computation on $G_{\lp}$ can efficiently
be simulated in the actual network.

In the context of local computation, each node $v\in V$ has to
independently decide whether it joins a vertex cover or dominating
set, which of its incident edges should participate in a matching,
or what variable its corresponding variable gets assigned when
solving an LP. Based on local knowledge, the nodes thus seek to
produce a feasible \emph{approximation} to the global optimization
problem. Depending on the number of rounds nodes communicate---and
thus on the amount of local knowledge available at the nodes---, the
quality of the solution that can be computed differs. We seek to
understand the trade-off between the amount of local knowledge (or
communication between nodes) and the resulting approximation to the
global problem.

In addition to these optimization problems, we also consider
important binary problems, including:
\begin{itemize}
\item \textbf{Maximal Independent Set (MIS):} Given a graph $G=(V,E)$,
  select an inclusion-maximal vertex subset $S\subseteq V$, such that
  no two nodes in $S$ are neighbors.
\item \textbf{Maximal Matching (MM):} Given a $G=(V,E)$, select an
  inclusion-maximal edge subset $T\subseteq E$, such that no two edges
  in $T$ are adjacent.
\end{itemize}
For such problems, we are interested in the question, how much local
information is required such that distributed decision makers are able
to compute fundamental graph-theoretic structures, such as an MIS or
an MM. Whereas most of the described combinatorial optimization
problems are NP-hard and thus, unless $\mathrm{P}=\mathrm{NP}$, even
with global knowledge, algorithms can compute only approximations to
the optimum, an MIS or an MM can trivially be computed with global
knowledge. The question is thus how much \emph{local} knowledge is
required to solve these tasks.


%% file: contribution.tex
\subsection{Contributions}

Our main results are a lower bound on the distributed
approximability of the minimum vertex cover problem in \Cref{sec:lower} as well as a generic algorithm for covering and
packing LPs of the form (P) and (D) in \Cref{sec:upper},
respectively. Both results are accompanied by various extensions and
adaptations to the other problems introduced in \Cref{sec:problems}. It follows from our discussion that these
results imply strong lower and upper bounds on the amount of local
information required to solve/approximate global tasks.

For the MVC lower bound, we show that for every $k>0$, there exists a
graph $G$ such that every $k$-round distributed algorithm for the MVC
problem has approximation ratios at least
\[
\Omega\left(\frac{n^{c/k^2}}{k}\right)\;\ \text{and}\quad
\Omega\left(\frac{\Delta^{1/(k+1)}}{k}\right)
\]
for a positive constant $c$, where $n$ and $\Delta$ denote
the number of nodes and the highest degree of $G$, respectively.
Choosing $k$ appropriately, this implies that to achieve
a constant approximation ratio, every MVC
algorithm requires at least $\Omega\big(\sqrt{\log n/\log\log n}\big)$ and
$\Omega\big(\log\Delta/\log\log \Delta\big)$ rounds,
respectively. All bounds also hold for randomized algorithms. Using
reductions that preserve the locality properties of the considered
graph, we show that the same lower bounds also hold for the
distributed approximation of the minimum dominating set and maximum
matching problems. Because MVC and MaxM are covering and packing
problems with constant integrality gap, the same lower bounds are
also true for general distributed covering and packing LPs of the
form (P) and (D). Furthermore, using locality-preserving reductions,
we also derive lower bounds on the amount of local information
required at each node to collectively compute important structures
such as an MIS or a maximal matching in the network graph.
Finally, a simple girth argument can be used to show that for the
connected dominating set problem, even stronger lower bounds are
true. We show that in $k$ rounds, no algorithm can have an
approximation ratio that is better than $n^{c/k}$ for some positive
constant $c$. This implies that for a polylogarithmic approximation
ratio, $\Omega(\log(n)/\log\log(n))$ rounds are needed.

We show that the above lower bound results that depend on $\Delta$ are
asymptotically almost tight for the MVC and
MaxM problem by giving an algorithm that obtains $\Oh(\Delta^{c/k})$
approximations with $k$ hops of information for a positive constant
$c$. That is, a constant approximation to MVC can be computed with
every node having $\Oh(\log\Delta)$-hop information and any
  polylogarithmic approximation ratio can be achieved in
  $\Oh(\log\Delta/\log\log\Delta)$ rounds. In recent work, it has been
shown that also a constant approximation can be obtained in time
$\Oh(\log\Delta/\log\log\Delta)$ and thus as a function of $\Delta$,
our MVC lower bound is also tight for contant approximation ratios \cite{baryehuda16}. Our main
upper bound result is a distributed algorithm to solve general
covering and packing LPs of the form (P) and (D). We show that with
$k$ hops of information, again for some positive constant $c$, a
$n^{c/k}$-approximation can be computed. As a consequence, by
choosing $k$ large enough, we also get a distributed approximation
scheme for this class of problems. For $\eps>0$, the algorithm
allows to compute an $(1+\eps)$-approximation in $\Oh(\log(n)/\eps)$
rounds of communication. Using a distributed randomized rounding
scheme, good solutions to fractional covering and packing problems
can be converted into good integer solutions in many cases. In
particular, we obtain the currently best distributed dominating set
algorithm, which achieves a $(1+\eps)\ln\Delta$-approximation for
MDS in $\Oh(\log(n)/\eps)$ rounds for $\eps>0$. Finally, we extend
the MDS result to connected dominating sets and show that up to
constant factors in approximation ratio and time complexity, we can
achieve the same time-approximation trade-off as for the MDS problem
also for the CDS problem.


%% file: relatedwork.tex
\section{Related Work}\label{sec:relwork}

\textbf{Local Computation: }Local algorithms have first been studied
in the Mid-1980s \cite{luby86,cole86}. The basic motivation was the
question whether one can build efficient network algorithms, where
each node only knows about its immediate neighborhood. However, even
today, relatively little is known about the fundamental limitations
of local computability. Similarly, little is known about \emph{local
approximability}, i.e., how well combinatorial optimization problems
can be approximated if each node has to decide individually based
only on knowledge available in its neighborhood.

Linial's seminal $\Omega(\log^*\!n)$ time lower bound for
constructing a maximal independent set on a ring~\cite{linial92} is
virtually the only non-trivial lower bound for local
computation.\footnote{There are of course numerous lower bounds and
impossibility results in distributed computing~\cite{fich03}, but
they apply to computational models where locality is not the key
issue. Instead, the restrictive factors are usually aspects such as
bounded message size~\cite{elkin-stoc04,DasSarma2012}, asynchrony, or faulty
processors.} Linial's lower bound shows that the non-uniform
$O(\log^* \! n)$ coloring algorithm by Cole and Vishkin
\cite{cole86} is asymptotically optimal for the ring. It has
recently been extended to other problems
\cite{czygrinow08fast,lenzen08leveraging}. On the other hand, it was
later shown that there exist non-trivial problems that can indeed be
computed \emph{strictly locally}. Specifically, Naor and Stockmeyer
present locally checkable labelings which can be computed in
constant time, i.e., with purely local information~\cite{naor93}.

There has also been significant work on (parallel) algorithms for
approximating packing and covering problems that are faster than
interior-point methods that can be applied to general LPs
(e.g.~\cite{fleischer00,plotkin95,young-focs01}). However, these
algorithms are not local as they need at least some global
information to work.\footnote{In general, a local algorithm provides
an efficient algorithm in the PRAM model of parallel computing, but
a PRAM algorithm is not necessarily local
~\cite{wattenhofer-disc04}.} The problem of approximating positive
LPs using only local information has been introduced
in~\cite{papadimitriou-podc91,papadimitriou-stoc93}. The first
algorithm achieving a constant approximation for general covering
and packing problems in polylogarithmic time is described in
\cite{bartal97}. Distributed (approximation) algorithms targeted for
specific covering and packing problems include algorithms for the
minimum dominating set problem
\cite{dubhashi03,jia-podc01,rajagopalan98,kuhn-podc03} as well as
algorithms for maximal matchings and maximal independent sets
\cite{alon86,israeli86,luby86}. We also refer to the survey
in~\cite{elkin04survey}.

While local computation was always considered an interesting and
elegant research question, 
several new application domains, such as overlay
or sensor networks, have reignited the attention to the area. 
Partly driven by these new application domains,
and partly due to the lower bounds presented in this paper, research
in the last five years has concentrated on restricted graph
topologies, such as unit disk graphs, bounded-growth graphs, or planar graphs. A
survey covering this more recent work is~\cite{suomela11survey}.

\textbf{Self-Organization \& Fault-Tolerance: }Looking at the wider
picture, one may argue that local algorithms even go back to the
early 1970s when Dijkstra introduced the concept of
\emph{self-stabilization}
\cite{dijkstra73self-stabilization,dijkstra74self-stabilizing}. A
self-stabilizing system must survive arbitrary failures, including
for instance a total wipe out of volatile memory at all nodes. The
system must self-heal and eventually converge to a correct state
from any arbitrary starting state, provided that no further faults
occur.

It seems that the world of self-stabilization (which is
asynchronous, long-lived, and full of malicious failures) has
nothing in common with the world of local algorithms (which is
synchronous, one-shot, and free of failures). However, as shown 20
years ago, this perception is
incorrect~\cite{awerbuch88dynamic,AfekKY90,awerbuch91distributed};
indeed it can easily be shown that the two areas are related.
Intuitively, this is because (i)~asynchronous systems can be made
synchronous, (ii)~self-stabilization concentrates on the case after
the last failure, when all parts of the system are correct again,
and (iii)~one-shot algorithms can just be executed in an infinite
loop. Thus, efficient self-stabilization essentially boils down to
local algorithms and hence, local algorithms are the key to
understanding fault-tolerance \cite{SSS2009}.

Likewise, local algorithms help to understand \emph{dynamic
networks}, in which the topology of the system is constantly
changing, either because of churn (nodes constantly joining or
leaving as in peer-to-peer systems), mobility (edge changes because
of mobile nodes in mobile networks), changing environmental
conditions (edge changes in wireless networks), or algorithmic
dynamics (edge changes because of algorithmic decisions in overlay
networks). In dynamic networks, no node in the network is capable of
keeping up-to-date global information on the network. Instead, nodes
have to perform their intended (global) task based on local
information only. In other words, all computation in these systems
is inherently local! By using local algorithms, it is guaranteed
that dynamics only affect a restricted neighborhood.
 Indeed, to the best of our knowledge, local algorithms yield the best solutions when it comes to
  dynamics. Dynamics also play a natural role in the area of self-assembly (DNA
computing, self-assembling robots, shape-shifting systems, or
claytronics), and as such it is not surprising that local algorithms
are being considered a key to understanding self-assembling systems
\cite{sterling09,claytronics}.

\textbf{Other Applications: }Local computation has also been
considered in a non-distributed (sequential) context. One example
are \emph{sublinear time algorithms}, i.e., algorithms that cannot
read the entire input, but must give (estimative) answers based on
samples only. For example, the local algorithms given in
\Cref{sec:upper} are used by Parnas and Ron \cite{parnas07}
to design a sublinear- or even constant-time sequential
approximation algorithms. In some sense the local algorithm plays
the role of an oracle that will be queried by random sampling, see
also \cite{nguyen08constant-time}.

There has recently been significant interest in the database community about 
the Pregel system~\cite{pregel}, a practically implemented computational model suitable for computing problems in large graphs. All our lower bounds directly apply to Pregel, i.e., they show how many iterations are required to solve certain tasks; while our upper bounds provide optimal or near-optimal algorithms in a Pregel-like message-passing system.

Finally, the term ``local(ity)'' is used in various different
contexts in computer science. The most common use may be
\emph{locality of reference} in software engineering. The basic idea
is that data and variables that are frequently accessed together
should also be physically stored together in order to facilitate
techniques such as caching and pre-fetching. At first glance, our
definition of locality does not seem to be related at all with
locality in software engineering. However, such a conclusion may be
premature. One may for instance consider a multi-core system where
different threads operate on different parts of data, and sometimes
share data. Two threads should never manipulate the same data at the
same time, as this may cause inconsistencies. At runtime, threads
may figure out whether they have conflicts with other threads,
however, there is no ``global picture''. One may model such a
multi-thread system with a virtual graph, with threads being nodes,
and two threads having a conflict by an edge between the two nodes.
Again, local algorithms (in particular maximal independent set or
vertex coloring) might help to efficiently schedule threads in a
non-conflicting way. At this stage, this is mostly a theoretical
vision \cite{schneider09isaac}, but with the rapid growth of
multi-core systems, it may get practical sooner than expected.


%% file: lowerbound.tex
\newcommand{\tree}[1]{\left[#1\right]}
\newcommand{\trmul}{\cdot}
\newcommand{\tradd}{\uplus}
\newcommand{\bigtradd}{\biguplus}

\section{Local Computation: Lower Bound}
\label{sec:lower}

The proofs of our lower bounds are based on the \emph{timeless
indistinguishability} argument \cite{fischer85,lamport82}. In $k$
rounds of communication, a network node can only gather
information about nodes which are at most $k$ hops away and hence,
only this information can be used to determine the computation's
outcome. If we can show that within their $k$-hop neighborhood many
nodes see exactly the same graph topology; informally speaking,
all these nodes are equally qualified to join the MIS, dominating
set, or vertex cover. The challenge is now to construct the graph
in such a way that selecting the wrong subset of these nodes is
ruinous.

We first construct a hard graph for the MVC problem because i) it has
a particularly simple combinatorial structure, and ii) it appears to
be an ideal candidate for local computation. At least when only
requiring relatively loose approximation guarantees, intuitively, a
node should be able to decide whether or not to join the vertex cover
using information from its local neighborhood only; very distant nodes
appear to be superfluous for its decision. Our proof shows that this
intuition is misleading and even such a seemingly simple problem such
as approximating MVC is not purely local; it cannot be approximated
well in a constant number of communication rounds. Our \emph{hardness
  of distributed approximation} lower bounds for MVC holds even for
\emph{randomized algorithms} as well as for the \emph{fractional}
version of MVC. We extend the result to other problems in
\Cref{sec:reductions}.

\textbf{Proof Outline: }The basic idea is to
construct a graph $G_k=(V,E)$, for each positive integer $k$. In
$G_k$, there are many neighboring nodes that see exactly the same
topology in their $k$-hop neighborhood, that is, no distributed
algorithm with running time at most $k$ can distinguish between
these nodes. Informally speaking, both neighbors are equally
qualified to join the vertex cover. However, choosing the wrong
neighbors in $G_k$ will be ruinous.

$G_k$ contains a bipartite subgraph $S$ with node set $C_0\cup
C_1$ and edges in $C_0\times C_1$ as shown in \Cref{fig:ct1}. Set $C_0$ consists of $n_0$ nodes each of which has
$\delta_0$ neighbors in $C_1$. Each of the
$n_0\cdot\frac{\delta_0}{\delta_1}$ nodes in $C_1$ has $\delta_1$,
$\delta_1 > \delta_0$, neighbors in $C_0$. The goal is to
construct $G_k$ in such a way that all nodes in $v\in S$ see the
same topology $\mathcal{T}_{v,k}$ within distance $k$. In a
globally optimal solution, all edges of $S$ may be covered by
nodes in $C_1$ and hence, no node in $C_0$ needs to join the
vertex cover. In a local algorithm, however, the decision of
whether or not a node joins the vertex cover depends only on its
local view, that is, the pair
$(\mathcal{T}_{v,k},\mathcal{L}(\mathcal{T}_{v,k}))$. We show that
because adjacent nodes in $S$ see the same $\mathcal{T}_{v,k}$,
every algorithm adds a large portion of nodes in $C_0$ to its
vertex cover in order to end up with a feasible solution. 
This yields suboptimal local decisions and hence, a suboptimal
approximation ratio. Throughout the proof, $C_0$ and $C_1$ denote
the two sets of the bipartite subgraph $S$.

The proof is organized as follows. The structure of $G_k$ is
defined in \Cref{sec:graph}. In
\Cref{sec:construction}, we show how $G_k$ can be
constructed without small cycles, ensuring that each node sees a
tree within distance $k$. \Cref{sec:equality} proves
that adjacent nodes in $C_0$ and $C_1$ have the same view
$\mathcal{T}_{v,k}$ and finally, \Cref{sec:analysis}
derives the local approximability lower bounds.

\subsection{The Cluster Tree}\label{sec:graph}
The nodes of graph $G_k=(V,E)$ can be grouped into disjoint sets
which are linked to each other as bipartite graphs. We call these
disjoint sets of nodes \emph{clusters}. The structure of $G_k$ is
defined using a directed tree $CT_k=(\mathcal{C},\mathcal{A})$
with doubly labeled arcs
$\ell:\mathcal{A}\rightarrow\mathbb{N}\times\mathbb{N}$. We refer
to $CT_k$ as the \emph{cluster tree}, because each vertex $C\in
\mathcal{C}$ represents a cluster of nodes in $G_k$. The
\emph{size} of a cluster $|C|$ is the number of nodes the cluster
contains. An arc $a=(C,D) \in \mathcal{A}$ with
$\ell(a)=(\delta_C,\delta_D)$ denotes that the clusters $C$ and
$D$ are linked as a bipartite graph, such that each node $u \in C$
has $\delta_C$ neighbors in $D$ and each node $v\in D$ has
$\delta_D$ neighbors in $C$. It follows that $|C|\cdot\delta_C =
|D|\cdot\delta_D$. We call a cluster \emph{leaf-cluster} if it is
adjacent to only one other cluster, and we call it
\emph{inner-cluster} otherwise.

\begin{definition}\label{defn:clustertree}
The cluster tree $CT_k$ is recursively defined as follows:
  \begin{eqnarray*}
    CT_1\!\!& := &\!\!(\mathcal{C}_1, \mathcal{A}_1), \quad
    \mathcal{C}_1\ :=\  \{C_0,C_1,C_2,C_3\} \\
    \mathcal{A}_1\!\!& := &\!\!\{(C_0,C_1),(C_0,C_2),(C_1,C_3)\} \\
    \ell(C_0,C_1)\!\!& := &\!\!(\delta_0,\delta_1), \quad
    \ell(C_0,C_2)\ :=\ (\delta_1,\delta_2), \\
    \ell(C_1,C_3)\!\!& := &\!\!(\delta_0,\delta_1)
  \end{eqnarray*}
Given $CT_{k-1}$, we obtain $CT_k$ in two steps:
  \begin{itemize}
    \item For each inner-cluster $C_i$, add a new leaf-cluster $C'_i$ with
    $\ell(C_i,C'_i):=(\delta_k,\delta_{k+1})$.
  \item For each leaf-cluster $C_i$ of $CT_{k-1}$ with
    $(C_{i'},C_i)\in\mathcal{A}$ and
    $\ell(C_{i'},C_i)=(\delta_p,\delta_{p+1})$, add $k\!-\!1$ new
    leaf-clusters $C'_j$ with
    $\ell(C_i,C'_j):=(\delta_j,\delta_{j+1})$ for $j=0\ldots k, j\neq
    p+1$.
\end{itemize}
Further, we define $|C_0|=n_0$ for all $CT_k$.
\end{definition}

\begin{figure}[t]
  \begin{center}
    \includegraphics[width=0.9\columnwidth]{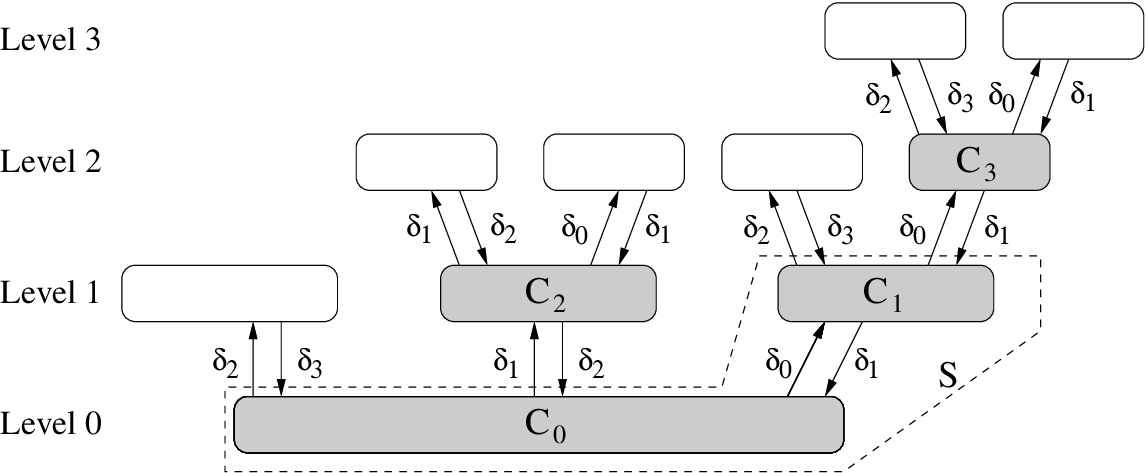}
    \caption{Cluster-Tree $CT_2$.}\label{fig:ct1}
  \end{center}
\end{figure}

\Cref{fig:ct1} shows $CT_2$. The shaded subgraph corresponds
to $CT_1$. The labels of each arc $a\in\mathcal{A}$ are of the
form $\ell(a)=(\delta_l,\delta_{l+1})$ for some $l\in \{0,\ldots
,k\}$. Further, setting $|C_0|=n_0$ uniquely determines the size
of all other clusters. In order to simplify the upcoming study of
the cluster tree, we need two additional definitions. The
\emph{level} of a cluster is the distance to $C_0$ in the cluster
tree (cf.~\Cref{fig:ct1}). The \emph{depth} of a cluster $C$
is its distance to the furthest leaf in the subtree rooted at $C$.
Hence, the depth of a cluster plus one equals the height of the
subtree corresponding to $C$. In the example of
\Cref{fig:ct1}, the depths of $C_0$, $C_1$, $C_2$, and $C_3$
are $3$, $2$, $1$, and $1$, respectively.

Note that $CT_k$ describes the general structure of $G_k$, i.e., it
defines for each node the number of neighbors in each cluster.
However, $CT_k$ does not specify the actual adjacencies. In the
next subsection, we show that $G_k$ can be constructed so that
each node's local view is a tree.

\subsection{The Lower-Bound Graph}\label{sec:construction}

In \Cref{sec:equality}, we will prove that the topologies
seen by nodes in $C_0$ and $C_1$ are identical. This task is greatly
simplified if each node's topology is a tree (rather than a general
graph) because we do not have to worry about cycles. The \emph{girth}
of a graph $G$, denoted by $g(G)$, is the length of the shortest cycle
in $G$. In the following, we show that it is possible to construct
$G_k$ with girth at least $2k+1$ so that in $k$ communication rounds,
all nodes see a tree.\footnote{The high-girth construction we use in
  this paper is based on the notion of graph lifts. For the original
  proof in~\cite{kuhn-podc04}, we used an alternative method based on
  a bipartite graph family of high girth developed by Lazebnik and
  Ustimenko~\cite{lazebnik95a}). Both techniques yield equivalent
  results, but the construction using graph lifts is easier.}

For the construction of $G_k$, we start with an arbitrary instance
$G_k'$ of the cluster tree which may have the minimum possible
girth~$4$. An elaboration of the construction of $G_k'$ is deferred to
\Cref{sec:analysis}. For now, we simply assume that $G_k'$
exists and we show how to use it to obtain $G_k$.  We start with some
basic definitions.  For a graph $H=(W,F)$, a graph
$\tilde{H}=(\tilde{W},\tilde{F})$ is called a \emph{lift} of $H$ if
there exists a \emph{covering map} from $\tilde{H}$ to $H$. A
covering map from $\tilde{H}$ to $H$ is a graph homomorphism
$\varphi:\tilde{W}\to W$ that maps each $1$-neighborhood in
$\tilde{H}$ to a $1$-neighborhood in $H$. That is, for each $v_0\in
\tilde{W}$ with neighbors $v_1,\dots,v_d\in \tilde{W}$, the neighbors
of $\varphi(v_0)$ in $W$ are $\varphi(v_1),\dots,\varphi(v_d)$ (such
that $\varphi(v_i)\neq\varphi(v_j)$ for $i\neq j$). Observe that given
a graph $G_k'$ that satisfies the specification given in
\Cref{sec:graph}, any lift $G_k$ of $G_k'$ also satisfies
the cluster tree specification.  In order to show that $G_k$ can be
constructed with large girth, it therefore suffices to show that there
exists a lift $\tilde{G}_k'$ of $G_k'$ such that $\tilde{G}_k'$ has large girth. In
fact, we will see that for every graph $H$, there exists a lift
$\tilde{H}$ such that $\tilde{H}$ has large girth (and such that the
size of $\tilde{H}$ is not too large). We start with two simple
observations.

\begin{lemma}\label{lemma:subgraphlift}
  Let $H=(W,F)$ be a graph and assume that $H'$ is a subgraph of $H$
  and $\tilde{H}$ is a lift of $H$. Then, there exists a lift
  $\tilde{H}'$ of $H'$ such that $\tilde{H'}$ is a subgraph of $\tilde{H}$.
\end{lemma}
\begin{proof}
  Let $\varphi$ be a covering map from $\tilde{H}$ to $H$. We construct
  $\tilde{H}'$ in the straightforward way. For every node $x\in
  V(\tilde{H})$, we add node $x$ to the node set $V(\tilde{H}')$ of
  $\tilde{H}'$ if and only if $\varphi(x)$ is a node of $H'$. Further,
  for every edge $\set{x,y}\in E(\tilde{H})$, we add $\set{x,y}$ as an
  edge to graph $\tilde{H}'$ if and only if $x\in V(\tilde{H})$, $y\in
  V(\tilde{H})$, and $\set{\varphi(x),\varphi(y)}$ is an edge of $H'$.
\hspace*{\fill}\end{proof}

\begin{lemma}\label{lemma:liftsgirth}
  Let $H=(W,F)$ be a graph and assume that
  $\tilde{H}=(\tilde{W},\tilde{F})$ is a lift of $H$. Then, the girth
  of $\tilde{H}$ is at least as large as the girth of $H$.
\end{lemma}
\begin{proof}
  Consider any cycle $\tilde{C}=(x_0,x_2,\dots,x_{\ell-1})$ of $\tilde{H}$
  (that is, for $i\in\set{0,\dots,\ell-1}$, $\set{x_i,x_{(i+1) \!\! \mod \ell}}$
  is an edge of $\tilde{H}$). Let $\varphi$ be a covering map from
  $\tilde{H}$ to $H$. Because $\varphi$ is a covering map, the nodes
  $\varphi(x_0),\varphi(x_1),\dots,\varphi(x_{\ell-1}),\varphi(x_0)$ form a closed
  walk of length $\ell$ on $H$. Therefore, the cycle $\tilde{C}$
  induces a cycle $C$ in $H$ of length at most $\ell$.
\hspace*{\fill}\end{proof}

We further use the following three existing results.

\begin{lemma}\label{lemma:regularsupergraph}\cite{regularsupergraph}
  Let $H=(W,F)$ be a simple graph and assume that $\Delta(H)$ is the
  largest degree of $H$. Then, there exists a simple
  $\Delta(H)$-regular graph $H'$ such that $H$ is a subgraph
  of $H'$ and $|V(H')| \leq |V(H)|+\Delta(H)+2$.
\end{lemma}

\begin{lemma}\label{lemma:girthgraphs}\cite{erdossachs}
  For any $d\geq 3$ and any $g\geq 3$, there exist $d$-regular
  (simple) graphs with girth at least $g$ and $d^{(1+o(1))g}$ nodes.
\end{lemma}

\begin{lemma}\label{lemma:commonlifts}\cite{angluin81}
  Consider some integer $d\geq 1$ and let $H_1=(W_1,F_1)$ and
  $H_2=(W_2,F_2)$ be two $d$-regular graphs. Then, there exists a
  graph $\tilde{H}=(\tilde{W},\tilde{F})$ such that $\tilde{H}$ is a
  lift of $H_1$ and a lift of $H_2$, and such that the number of nodes
  of $\tilde{H}$ is at most $|V(\tilde{H})|\leq 4\cdot
  |V(H_1)|\cdot|V(H_2)|$.
\end{lemma}

Combining the above lemmas, we have all the tools needed to construct
$G_k$ with large girth as summarized in the following lemma.

\begin{lemma}\label{lemma:largegirth}
  Assume that for given parameters $k$ and
  $\delta_0,\dots,\delta_{k+1}$, there exists an instance $G_k'$ of
  the cluster tree with $n'$ nodes. Then, for every $g\geq 4$, there
  exists an instance $G_k$ of the cluster tree with girth $g$ and
  $O(n'\cdot \Delta^{(1+o(1))g})$ nodes, where $\Delta$ is the maximum
  degree of $G_k'$ (and hence also of $G_k$).
\end{lemma}
\begin{proof}
  We consider the instance $G_k'$ of the cluster tree. Consider any
  lift $\tilde{G}_k'$ of $G_k'$ and let $\varphi$ be a covering map from
  $\tilde{G}_k'$ to $G_k'$. If $G_k'$ follows the cluster tree structure
  given in \Cref{defn:clustertree}, $\tilde{G}_k'$ also follows the
  structure for the same parameters $k$ and
  $\delta_0,\dots,\delta_{k+1}$. To see this, given a cluster $C'$ of
  $G_k'$, we define the corresponding cluster $\tilde{C}'$ of $\tilde{G}_k'$ to
  contain all the nodes of $\tilde{G}_k'$ that are mapped into $C'$ by
  $\varphi$. The graph $\tilde{G}_k'$ then satisfies \Cref{defn:clustertree} (each node has the right number of neighbors
  in neighboring clusters) because $\varphi$ is a covering map. To
  prove the lemma, it is therefore sufficient to show that there
  exists a lift $\tilde{G}_k'$ of $G_k'$ such that $\tilde{G}_k'$ has girth at least $g$
  and such that $\tilde{G}_k'$ has at most $O(n'\Delta^{(1+o(1))g})$ nodes.

  We obtain such a lift $\tilde{G}_k'$ of $G_k'$ by applying the above
  lemmas in the following way. First of all, by 
  \Cref{lemma:regularsupergraph}, there exists a $\Delta$-regular
  supergraph $\bar{G}_k'$ of $G_k'$ with $O(n')$ nodes. Further, by
  \Cref{lemma:girthgraphs}, there exists a $\Delta$-regular graph
  $H$ with girth at least $g$ and $\Delta^{g(1+o(1))}$ nodes. Using
  \Cref{lemma:commonlifts}, there exists a common lift
  $\tilde{\bar{G}}_k'$ of $\bar{G}_k'$ and $H$ such that
  $\tilde{\bar{G}}_k'$ has at most
  $4|V(\bar{G}_k')||V(H)|=O(n'\Delta^{g(1+o(1))})$ nodes. By \Cref{lemma:liftsgirth}, because $\tilde{\bar{G}}_k'$ is a lift of
  $H$, the girth of $\tilde{\bar{G}}_k'$ is at least $g$. Further,
  because $\tilde{\bar{G}}_k'$ is a lift of $\bar{G}_k'$ and because
  $G_k'$ is a subgraph of $\bar{G}_k'$, by \Cref{lemma:subgraphlift} there exists a subgraph $\tilde{G}_k'$ of
  $\tilde{\bar{G}}_k'$ such that $\tilde{G}_k'$ is a lift of $G_k'$,
  which proves the claim of the lemma.  
\hspace*{\fill}\end{proof}

\subsection{Equality of Views}\label{sec:equality}
In this subsection, we prove that two adjacent nodes in clusters
$C_0$ and $C_1$ have the same \emph{view}, i.e., within distance
$k$, they see exactly the same topology $\mathcal{T}_{v,k}$.
Consider a node $v \in G_k$. Given that $v$'s view is a tree, we
can derive its \emph{view-tree} by recursively following all
neighbors of $v$. The proof is largely based on the observation
that corresponding subtrees occur in both node's view-tree.

Let $C_i$ and $C_j$ be adjacent clusters in $CT_k$ connected by
$\func{\ell}{C_i,C_j}=(\delta_l,\delta_{l+1})$, i.e., each node in
$C_i$ has $\delta_l$ neighbors in $C_j$, and each node in $C_j$
has $\delta_{l+1}$ neighbors in $C_i$. When traversing a node's
view-tree, we say that we \emph{enter} cluster $C_j$ (resp.,
$C_i$)  over \emph{link} $\delta_l$ (resp., $\delta_{l+1}$) from
cluster $C_i$ (resp., $C_j$).

Let $T_u$ be the tree topology seen by some node $u$ of degree $d$ and
let $T_1,\dots,T_d$ be the topologies of the subtrees of the $d$
neighbors of $u$. We use the following notation to describe the
topology $T_u$ based on $T_1,\dots,T_d$:
\[
T_u := \tree{T_1\tradd T_2\tradd\dots\tradd T_d} = \tree{\bigtradd_{i=1}^d T_i}.
\]
Further, we define the following abbreviation:
\[
d\trmul T := \underbrace{T\tradd T\tradd\dots\tradd T}_{d\text{ times}}.
\]
If $\mathcal{T}_1,\dots,\mathcal{T}_d$ are sets of trees, we use
$\tree{\mathcal{T}_1\tradd\dots\tradd\mathcal{T}_d}$ to denote the set
of all view-trees $\tree{T_1\tradd\dots\tradd T_d}$ for which
$T_i\in\mathcal{T}_i$ for all $i$. We will also mix single trees and
sets of trees, e.g., $\tree{\mathcal{T}_1\tradd T}$ is the set of all
view-trees $\tree{T_1\tradd T}$ with $T_1\in \mathcal{T}_1$.

\begin{definition}\label{defn:B}
  The following nomenclature refers to subtrees in the view-tree of
  a node in $G_k$.
  \begin{itemize}
  \item $\M{i}$ is the subtree seen upon
    entering cluster $C_0$ over a link $\delta_i$.
  \item $\Bup{i}{d}{\lambda}$ denotes the set of subtrees that
    are seen upon entering a cluster
    $C\in \mathcal{C} \setminus \{C_0\}$ on level $\lambda$ over a
    link $\delta_i$ from level $\lambda-1$, where
    $C$ has depth $d$.
  \item $\Bdown{i}{d}{\lambda}$ denotes the set of subtrees that
    are seen upon entering a cluster
    $C\in \mathcal{C} \setminus \{C_0\}$ on level $\lambda$ over a
    link $\delta_i$ from level $\lambda+1$, where
    $C$ has depth $d$.
 \end{itemize}
\end{definition}

In the following, we will frequently abuse notation and write
$\B{i}{d}{\lambda}$ when we mean some tree in $\B{i}{d}{\lambda}$.
The following example should clarify the various definitions.
Additionally, you may refer to the example of $G_3$ in Figure~7.2.
\begin{example}\label{ex:example}Consider $G_1$. Let $V_{C_0}$ and
  $V_{C_1}$ denote the view-trees of nodes in $C_0$ and $C_1$,
  respectively: \\
  \begin{align*}
    V_{C_0} & \in \tree{\delta_0\trmul\Bup{0}{1}{1} \tradd
      \delta_1\trmul\Bup{1}{0}{1}} &
    V_{C_1} & \in \tree{\delta_0\trmul\Bup{0}{0}{2} \tradd \delta_1\trmul \M{1}}\\
    \Bup{0}{1}{1}& \subseteq \tree{\delta_0\trmul\Bup{0}{0}{2} \tradd
      (\delta_1-1)\trmul \M{1}} &
    \Bup{0}{0}{2} & \subseteq \tree{(\delta_1-1)\trmul\Bdown{1}{1}{1}}\\
    \Bup{1}{0}{1} & \subseteq \tree{(\delta_2-1)\trmul \M{2}} &
    \M{1} & \in \tree{(\delta_0-1)\trmul\Bup{0}{1}{1} \tradd
      \delta_1\trmul\Bup{1}{0}{1}}\\
    \M{2} & \in \tree{\delta_0\trmul\Bup{0}{1}{1} \tradd
      (\delta_1-1)\trmul\Bup{1}{0}{1}} &
    & \dots
  \end{align*}
\end{example}
We start the proof by giving a set of rules which describe the
subtrees seen at a given point in the view-tree. We call these
rules \emph{derivation rules} because they allow us to
\emph{derive} the view-tree of a node by mechanically applying the
matching rule for a given subtree.

\begin{figure}[p]
\begin{center}
  \includegraphics[width=1\textwidth]{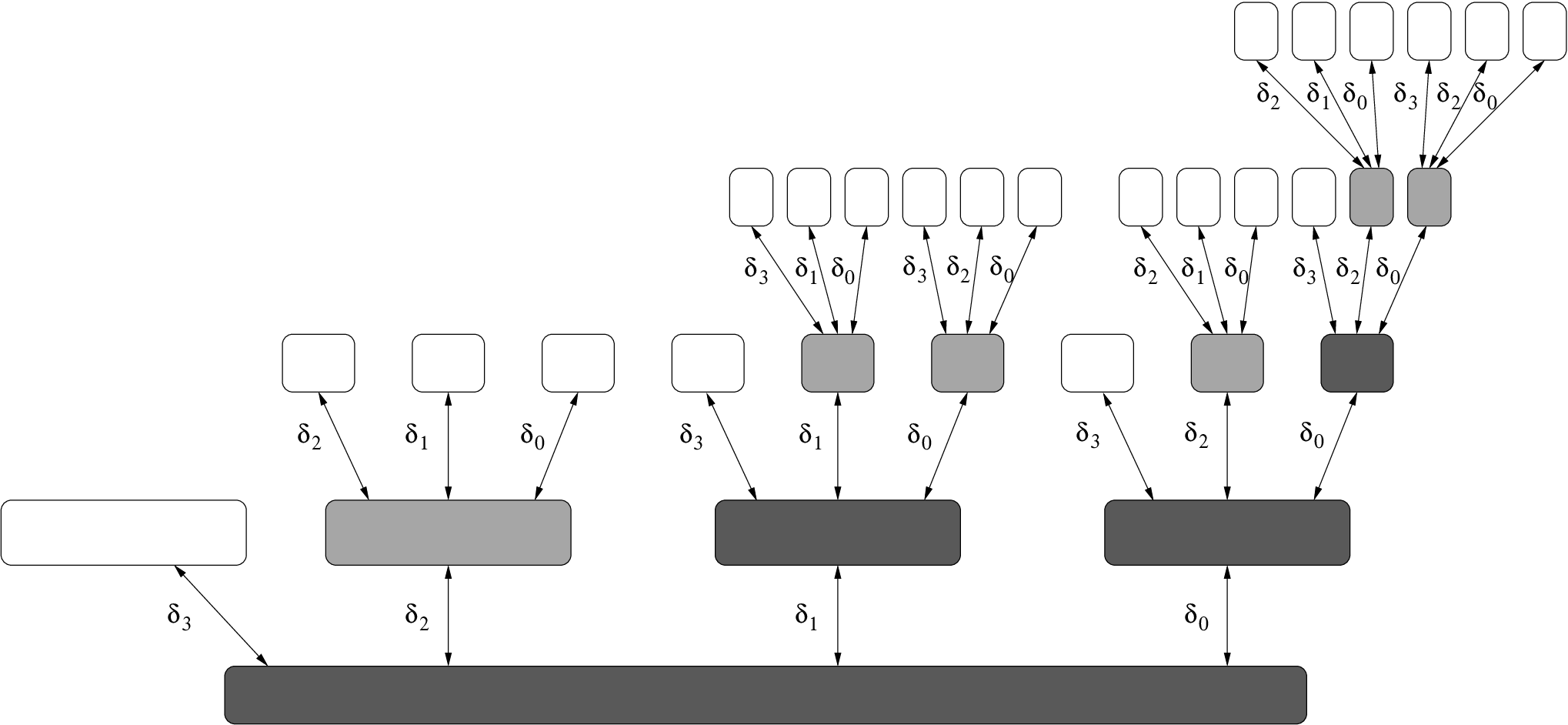}
  \\[12mm]
  \includegraphics[width=1\textwidth]{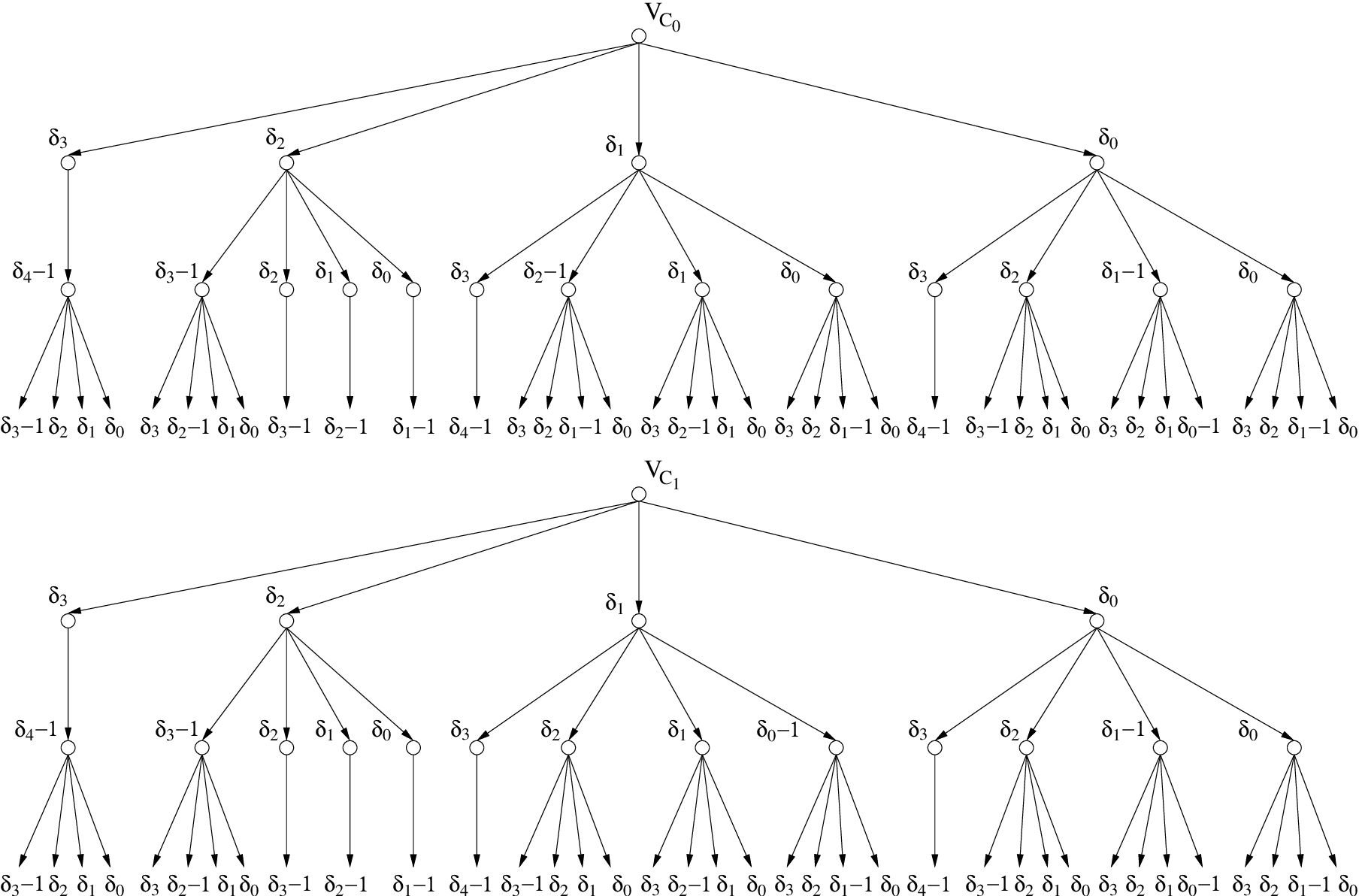}\label{fig:large}
  \\[6mm]
  \caption{The Cluster Tree $CT_3$ and the corresponding
    view-trees of nodes in $C_0$ and $C_1$. The cluster trees $CT_1$
    and $CT_2$ are shaded dark and light, respectively. The labels
    of the arcs of the cluster tree represent the number of
    higher-level cluster. The labels of the reverse links are
   omitted. In the view-trees, an arc labeled with $\delta_i$
    stands for $\delta_i$ edges, all connecting to identical
    subtrees.}
\end{center}
\end{figure}

\begin{lemma}\label{lemma:structure}
The following derivation rules hold in $G_k$:
\begin{align*}
    \M{i} & \in \tree{
      (\delta_{i-1}-1)\trmul\Bup{i-1}{k-i+1}{1} \tradd
      \bigtradd_{j\in\set{0,\dots,k}\setminus\set{i-1}}\!\!\!\!\!\!\!
      \delta_j\trmul\Bup{j}{k-j}{1}}\\
    \Bup{i}{d}{1} & \subseteq  \tree{\mathcal{F}_{\set{i+1},d,1} \tradd
      \mathcal{D}_{d,1} \tradd (\delta_{i+1}-1)\trmul \M{i+1}}\\
    \Bdown{i}{k-i}{1} & \subseteq \tree{\mathcal{F}_{\set{i-1,i+1},k-i,1} \tradd
      \mathcal{D}_{k-i,1} \tradd \delta_{i+1}\trmul \M{i+1} \tradd
      (\delta_{i-1}-1)\trmul \Bup{i-1}{k-i-1}{2}}\\
    \Bup{i-2}{k-i}{2} & \subseteq \tree{\mathcal{F}_{\set{i-1},k-i,2} \tradd
      \mathcal{D}_{k-i,2} \tradd
      (\delta_{i-1}-1)\trmul\Bdown{i-1}{k-i+1}{1}}
    \quad(i\geq2),
 \end{align*}
  where $\mathcal{F}$ and $\mathcal{D}$ are defined as
  \begin{align*}
    \mathcal{F}_{W,d,\lambda} & := \bigtradd_{j\in\set{0,\dots,k-d+1}\setminus W}\!\!\!\!\!
    \delta_j\trmul\Bup{j}{d-1}{\lambda+1}\\
    \mathcal{D}_{d,\lambda} & := \bigtradd_{j=k-d+2}^k
    \delta_j\trmul\Bup{j}{k-j}{\lambda+1}.
\end{align*}
\end{lemma}
\begin{proof}
  We first show the derivation rule for $\M{i}$. By \Cref{defn:B}, $\M{i}$ is the subtree seen upon entering the cluster
  $C_0$ over a link $\delta_i$. Let us therefore first derive a rule
  for the view-tree $V_{C_0}$ of $C_0$ in $G_k$.  We show by induction
  on $k$ that $V_{C_0} \in\tree{
    \bigtradd_{j\in\set{0,\dots,k}}\delta_j\trmul\Bup{j}{k-j}{1}}$. It
  can be seen in \Cref{ex:example} that the rule holds for
  $k=1$. For the induction step, we build $CT_{k+1}$ from $CT_{k}$ as
  defined in \Cref{defn:clustertree}. $C_0$ is an inner
  cluster and therefore, one new cluster with view trees of the form
  $\Bup{k+1}{0}{1}$ is added. The depth of all other subtrees
  increases by 1 and thus the rule for $V_{C_0}$ follows.  If we
  enter $C_0$ over link $\delta_i$, there will be only $\delta_{i-1}
  -1$ edges left to return to the cluster from which we had entered
  $C_0$. Consequently, $\M{i}$ is the same as $V_{C_0}$ but with only
  $\delta_{i-1}-1$ subtrees of the form $\Bup{i-1}{k-i+1}{1}$.

  The remaining rules follow along similar lines. Let $C_i$ be a
  cluster with entry-link $\delta_i$ which was first created in
  $CT_r$, $r<k$. Note that in $CT_k$, the depth of $C_i$ is $d=k-r$  because each
  subtree increases its depth by one in each ``round''. According to
  the second building rule of \Cref{defn:clustertree}, $r$
  new neighboring clusters (subtrees) are created in $CT_{r+1}$.  More
  precisely, a new cluster is created for all entry-links
  $\delta_0\ldots \delta_r$, except $\delta_i$. We call these subtrees
  \emph{fixed-depth} subtrees $F$. If the subtree with root $C_i$ has
  depth $d$ in $CT_k$, the fixed-depth subtrees have depth $d-1$. In
  each $CT_{r'}, \; r'\in\{r+2,\ldots ,k\}$, $C_i$ is an inner-cluster
  and hence, one new neighboring cluster with entry-link $\delta_{r'}$
  is created. We call these subtrees \emph{diminishing-depth} subtrees
  $D$. In $CT_k$, each of these subtrees has grown to depth $k-r'$.

  We now turn our attention to the differences between the three
  rules. They stem from the exceptional treatment of level $1$, as
  well as the predicates $\uparrow$ and $\downarrow$. In
  $\Bup{i}{d}{1}$, the link $\delta_{i+1}$ returns to $C_0$, but
  contains only $\delta_{i+1}-1$ edges in the view-tree.

  In $\Bdown{i}{k-i}{1}$, we have to consider two special cases. The
  first one is the link to $C_0$. For a cluster on level $1$ with
  depth $d$ and entry-link (from $C_0$) $\delta_j$, the equality
  $k=d+j$ holds and therefore, the link to $C_0$ is $\delta_{i+1}$ and
  thus, $\M{i+1}$ follows. Secondly, because in $\Bdown{i}{k-i}{1}$ we
  come from a cluster $C'$ on level $2$ over a $\delta_i$ link, there
  are only $\delta_{i-1}-1$ links back to $C'$ and therefore there are
  only $\delta_{i-1}-1$ subtrees of the form
  $\Bup{i-1}{k-i-1}{2}$. (Note that since we entered the current
  cluster from a higher level, the link leading back to where we came
  from is $\delta_{i-1}$, instead of $\delta_{i+1}$). The depths of
  the subtrees in $\Bdown{i}{k-i}{1}$ follow from the above
  observation that the view $\Bdown{i}{k-i}{1}$ corresponds to the
  subtree of a node in the cluster on level $1$ that is entered on
  link $\delta_{i+1}$ from $C_0$.

  Finally in $\Bup{i-2}{k-i}{2}$, we again have to treat the returning
  link $\delta_{i-1}$ to the cluster on level $1$ specially. Consider
  a view-tree $\Bup{j}{d}{x}$ of depth $d$. All level $x+1$ subtrees
  that are reached by links $\delta_{j'}$ with $j'\leq k-d+1$ are
  fixed-depth subtrees of depth $d-1$. Because $\Bup{i-2}{k-i}{2}$ is
  reached through link $\delta_{i-1}$ from the cluster on level $1$
  and because $i-1\leq k-(k-i)=i$, $\Bup{i-2}{k-i}{2}$ is a
  fixed-depth subtree of its level $1$ parent. Thus, we get that the depth
  of the $\delta_{i-1}-1$ subtrees $\Bdown{i-1}{k-i+1}{1}$ is $k-i+1$.
\hspace*{\fill}\end{proof}

Note that we do not give general derivation rules for all
$\Bup{i}{d}{x}$ and $\Bdown{i}{d}{x}$ because they are not needed in
the following proofs. Next, we define the notion of
\emph{$r$-equality}. Intuitively, if two view-trees are $r$-equal,
they have the same topology within distance $r$.

\begin{definition}\label{defn:requal}
  Let $\mathcal{V}\subseteq\tree{\bigtradd_{i=1}^d{\mathcal{T}_i}}$ and
  $\mathcal{V}'\subseteq\tree{\bigtradd_{i=1}^d{\mathcal{T}'_i}}$ be sets of
  view-trees. Then, $\mathcal{V}$ and $\mathcal{V}'$ are $r$-equal if
  there is a permutation $\pi$ on $\set{1,\dots,d}$ such that
  $\mathcal{T}_i$ and $\mathcal{T}'_{\pi(i)}$ are$(r-1)$-equal for all
  $i\in\set{1,\dots,d}$:
  \begin{equation}
    \mathcal{V} \requal{r} \mathcal{V}' \; \Longleftarrow \;
    \mathcal{T}_i \requal{r-1} \mathcal{T}'_{\pi(i)} \;,
    \; \forall i\in \{1,\ldots,d\}. \nonumber
  \end{equation}
  Further, all (sets of) subtrees are $0$-equal, i.e.,
  $\mathcal{T}\requal{0}\mathcal{T}'$ for all $\mathcal{T}$ and
  $\mathcal{T}'$.
\end{definition}

Using the notion of $r$-equality, we can now define what we actually
have to prove. We will show that in $G_k$, $V_{C_0} \requal{k}
V_{C_1}$ holds. Since the girth of $G_k$ is at least $2k+1$, this is
equivalent to showing that each node in $C_0$ sees exactly the same
topology within distance $k$ as its neighbor in $C_1$. We first
establish several helper lemmas. For two collections of sub-tree sets
$\beta=\mathcal{T}_1\tradd\dots\tradd\mathcal{T}_d$ and
$\beta'=\mathcal{T}'_1\tradd\dots\tradd\mathcal{T}'_d$, we say that
$\beta\requal{r}\beta'$ if $\mathcal{T}_i\requal{r}\mathcal{T}'_i$ for
all $i\in[d]$.

\begin{lemma}\label{lemma:equality}
  Let $\beta=\bigtradd_{i=1}^t\mathcal{T}_i$ and $\beta
  '=\bigtradd_{i=1}^t\mathcal{T}'_i$ be collections of sub-tree sets
  and  let
  \begin{equation*}
    V_{v} \in \tree{
      \beta \tradd
      \bigtradd_{i\in I}\delta_i\trmul\Bup{i}{d_i}{x_i}
    }
   \quad\text{and}\quad
    V_{v'} \in \tree{
      \beta' \tradd
      \bigtradd_{i\in I}\delta_i\trmul\Bup{i}{d_i'}{x_i'}
    }
  \end{equation*}
  for a set of integers $I$ and integers $d_i$, $d_i'$, $x_i$, and
  $x_i'$ for $i\in I$. Let $r\geq 0$ be an integer. If for all $i\in
  I$, $d_i=d_i'$ or $r\leq 1+\min\set{d_i,d_i'}$, it holds that
  \begin{equation*}
    \mathcal{V}_{v_1} \requal{r} \mathcal{V}_{v_2} \; \Longleftarrow  \; \beta
    \requal{r-1} \beta '.
  \end{equation*}
\end{lemma}
\begin{proof}
  Assume that the roots of the subtree of $V_{v}$ and $V_{v'}$ are in
  clusters $C$ and $C'$, respectively. W.l.o.g., we assume that
  $d'\leq d$. Note that we have $d'\geq 1+\min_{i\in
    I}\min\set{d_i,d_i'}$ and thus $r\leq d'$. In the construction
  process of $G_k$, $C$ and $C'$ have been created in steps $k-d$ and
  $k-d'$, respectively.

  By \Cref{defn:clustertree}, all subtrees with depth
  $d^*<d'$ have grown identically in both views $V_v$ and
  $V_{v'}$. The remaining subtrees of $V_{v'}$ were all created in
  step $k-d'+1$ and have depth $d'-1$. The corresponding subtrees in
  $V_{v}$ have at least the same depth and the same structure up to
  that depth. Hence paths of length at most $d'$ which start at the
  roots of $V_{v}$ and $V_{v}$, go into one of the subtrees in
  $\Bup{i_0}{d_{i_0}}{x_{i_0}}$ and $\Bup{i_0}{d_{i_0}'}{x_{i_0}'}$
  for $i_0\in I$ and do not return to clusters $C$ and $C'$ must be
  identical. Further, consider all paths which, after $s\leq d'$ hops, return
  to $C$ and $C$ over link $\delta_{i_0+1}$. After these $s$ hops,
  they return to the original cluster and see views
  \begin{align*}
    V'_{v} &\in \tree{
      \beta\tradd(\delta_{i_0}-1)\trmul\Bup{i_0}{d_{i_0}}{x_{i_0}}
      \tradd\bigtradd_{i\in I\setminus\set{i_0}}
      \delta_i\trmul\Bup{i}{d_i}{x_i}
    }
    \quad\text{and}\\
    V'_{v'} &\in \tree{
      \beta'\tradd(\delta_{i_0}-1)\trmul\Bup{i_0}{d_{i_0}'}{x_{i_0}'}
      \tradd\bigtradd_{i\in I\setminus\set{i_0}}
      \delta_i\trmul\Bup{i}{d_i'}{x_i'}
    },
  \end{align*}
  differing from $V_{v}$ and $V_{v'}$ only in having $\delta_{i_0}-1$
  instead of $\delta_{i_0}$ subtrees in $\Bup{i_0}{d_{i_0}}{x_{i_0}}$
  and $\Bup{i_0}{d_{i_0}'}{x_{i_0}'}$, respectively. This does not
  affect $\beta$ and $\beta '$ and therefore,
  \begin{equation}
    V_{v_1} \requal{r} V_{v_2} \;  \Longleftarrow V'_{v_1}\requal{r-s}
    \; V'_{v_2} \;\wedge \;\beta \;\requal{r-1}\; \beta ' \; ,\; s>1.
    \nonumber
  \end{equation}
  Note that $s\leq d'$ implies $s\leq r$. Thus, the same argument can
  be repeated until $r-s=0$ and because $V'_{v}\requal{0} \;
  V'_{v'}$, we can conclude that $V_{v}\requal{r}V_{v}$ and thus
  the lemma follows.
\hspace*{\fill}\end{proof}

\begin{figure}[ht]
  \begin{center}
    \includegraphics[width=0.65\columnwidth]{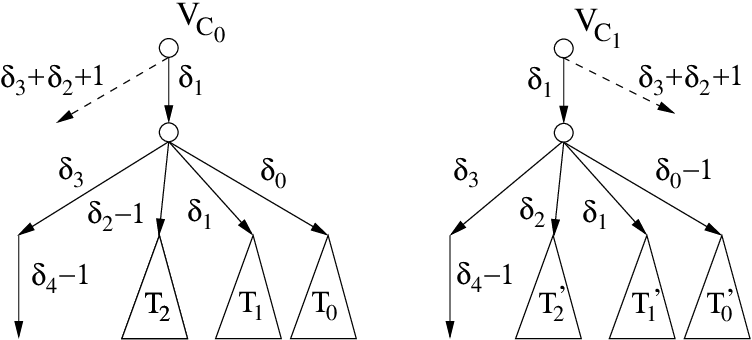}
    \caption{The view-trees $V_{C_0}$ and $V_{C_1}$ in $G_3$ seen upon
    using link $\delta_1$.}\label{fig:equality}
  \end{center}
\end{figure}

\Cref{fig:equality} shows a part of the view-trees of nodes
in $C_0$ and $C_1$ in $G_3$. The figure shows that the subtrees
with links $\delta_0$ and $\delta_2$ cannot be matched directly to
one another because of the different placement of the $-1$. It
turns out that this inherent difference appears in every step of
our theorem. However, the following lemma shows that the subtrees
$T_0$ and $T_2$ ($T'_0$ and $T'_2$) are equal up to the required
distance and hence, nodes are unable to distinguish them. It is
this crucial property of our cluster tree, which allows us to
``move'' the ``$-1$'' between links $\delta_i$ and
$\delta_{i+2}$ and enables us to derive the main theorem.

\begin{lemma}\label{lemma:changemin}
  For all $i\in\set{2,\dots,k}$, we have
  \[
  M_i \requal{k-i-1} \Bup{i-2}{k-i}{2}.
  \]
\end{lemma}
\begin{proof}
  By \Cref{lemma:structure}, we have
  \begin{align*}
    \M{i} & \in \tree{
      \bigtradd_{j\in\set{0,\dots,k}\setminus\set{i-1}}\!\!\!\!\!\delta_j\trmul\Bup{j}{k-j}{1}
      \tradd (\delta_{i-1}-1)\trmul\Bup{i-1}{k-i+1}{1}
    } \\
    \Bup{i-2}{k-i}{2} & \subseteq \tree{
      \bigtradd_{\substack{j\in\set{0,\dots,i+1}\\\setminus\set{i-1}}}\!\!\!\!\!\!\!\!\!\!
      \delta_j\trmul\Bup{j}{k-i-1}{3} \tradd
      \!\!\bigtradd_{j\in\set{i+2,\dots,k}}\!\!\!\!\!\!\!\!\!\!
      \delta_j\trmul\Bup{j}{k-j}{3} \tradd
      (\delta_{i-1}-1)\trmul\Bdown{i-1}{k-i+1}{1}
    }
  \end{align*}
  Consider the subtrees of the form $\Bup{j}{d}{x}$ for $j\neq i-1$ in
  both cases. For $j\leq i+1$, the depths of the subtrees are $k-j$
  and $k-i-1\leq k-j$, respectively. For $j>i+1$, the depths are $k-j$
  in both cases. We can therefore apply \Cref{lemma:equality} to
  get that
  \begin{equation}\label{eq:structurestep}
    M_i \requal{k-i-1} \Bup{i-2}{k-i}{2} \Longleftarrow
    \Bup{i-1}{k-i+1}{1} \requal{k-i-2} \Bdown{i-1}{k-i+1}{1}.
  \end{equation}
  To show that the right-hand side of \cref{eq:structurestep} holds,
  we plug $\Bup{i-1}{k-i+1}{1}$ and $\Bdown{i-1}{k-i+1}{1}$ into \Cref{lemma:structure} and use the derivation
  rules:
  \begin{align*}
    \Bup{i-1}{k-i+1}{1} & \subseteq  \tree{\mathcal{F}_{\set{i},k-i+1,1} \tradd
      \mathcal{D}_{k-i+1,1} \tradd (\delta_i-1)\trmul
      \M{i}}\\
    & = \tree{\mathcal{F}_{\set{i-2,i},k-i+1,1} \tradd
      \mathcal{D}_{k-i+1,1} \tradd (\delta_i-1)\trmul
      \M{i} \tradd \delta_{i-2}\trmul\Bup{i-2}{k-i}{2}}\\
    \Bdown{i-1}{k-i+1}{1} & \subseteq \tree{\mathcal{F}_{\set{i-2,i},k-i+1,1} \tradd
      \mathcal{D}_{k-i+1,1} \tradd \delta_{i}\trmul \M{i} \tradd
      (\delta_{i-2}-1)\trmul \Bup{i-2}{k-i}{2}}
  \end{align*}
  The two expressions are equal except for the placement of the
  ``$-1$''. Therefore, we get
  $\Bup{i-1}{k-i+1}{1}\requal{k-i-2}\Bdown{i-1}{k-i+1}{1}$ if
  $M_i\requal{k-i-3}\Bup{i-2}{k-i}{2}$. Hence, we have shown that
  \[
  \M{i} \requal{k-i-1} \Bup{i-2}{k-i}{2}
  \Longleftarrow
  \Bup{i-1}{k-i+1}{1} \requal{k-i-2} \Bdown{i-1}{k-i+1}{1}
  \Longleftarrow \M{i} \requal{k-i-3} \Bup{i-2}{k-i}{2}.
  \]

  This process can be continued using exactly the same rules until the
  requirement becomes that either
  \[
  \Bup{i-1}{k-i+1}{1} \requal{0} \Bdown{i-1}{k-i+1}{1}
  \quad\text{or}\quad
  \M{i} \requal{0} \Bup{i-2}{k-i}{2},
  \]
  which is always true.
\hspace*{\fill}\end{proof}

Finally, we are ready to prove the main theorem.

\begin{theorem}\label{thm:equality}
  Consider graph $G_k$. Let $V_{C_0}$ and $V_{C_1}$ be the view-trees
  of two adjacent nodes in clusters $C_0$ and $C_1$,
  respectively. Then, $V_{C_0} \requal{k} V_{C_1}$.
\end{theorem}
\begin{proof}
  By the construction of $G_k$, the view-trees of $V_{C_0}$ and
  $V_{C_1}$ can be written as
  \begin{align*}
    V_{V_0} & \in \tree{\bigtradd_{j=0}^k
      \delta_j\trmul\Bup{j}{k-j}{1}}\quad\text{and}\\
    V_{V_1} & \in \tree{\delta_1\trmul M_1 \tradd
      \bigtradd_{j\in\set{0,\dots,k}\setminus\set{1}}\delta_j\trmul\Bup{j}{k-j}{2}}.
  \end{align*}
  It follows that $V_{C_0} \requal{k} V_{C_1} \; \Longleftarrow \;
  \Bup{1}{k-1}{1} \requal{k-1} \M{1}$ by 
  \Cref{lemma:equality}. To prove that
  $\Bup{1}{k-1}{1}\requal{k-1}\M{1}$, we show that
  \begin{equation}\label{eq:equalityinduction}
    \Bup{k-s}{s}{1}\requal{s}\M{k-s}\quad\text{for all}\quad s\in\set{0,\dots,k-1}.
  \end{equation}
  We show \Cref{eq:equalityinduction} by induction on
  $s$. The statement is trivially true for $s=0$ because any two trees
  are $0$-equal. For the induction step, consider the derivation rules
  for $\Bup{k-s}{s}{1}$ and $\M{k-s}$:
  \begin{align*}
    \Bup{k-s}{s}{1} & \subseteq  \tree{
      \bigtradd_{j=0}^{k-s} \delta_j\trmul\Bup{j}{s-1}{2}\tradd
      \bigtradd_{j=k-s+2}^k \delta_j\trmul\Bup{j}{k-j}{2}
      \tradd (\delta_{k-s+1}-1)\trmul \M{k-s+1}
    }\\
    \M{k-s} & \in \tree{
      \bigtradd_{j\in\set{0,\dots,k}\setminus\set{k-s-1}} \delta_j\trmul\Bup{j}{k-j}{1}
      \tradd (\delta_{k-s-1}-1)\trmul \Bup{k-s-1}{s+1}{1}
    }
  \end{align*}
  Consider the subtrees of the form $\Bup{j}{d}{x}$ for
  $j\not\in\set{k-s-1,k-s+1}$ in both expressions. For $j\leq k-s$,
  the depths of the subtrees are $s-1< k-j$ and $k-j$,
  respectively. For $j>k-s+1$, the depth is $k-j$ in both
  cases. Hence, we can apply \Cref{lemma:equality} to get
  \begin{eqnarray}   \label{eq:equalityindstep}
    \lefteqn{\Bup{k-s}{s}{1} \requal{s} \M{k-s} \Longleftarrow}\\
    && a\trmul\Bup{k-s-1}{s+1}{2} \tradd
    (b-1)\trmul\M{k-s+1}
    \requal{s-1}
   (a-1)\trmul\Bup{k-s-1}{s+1}{1} \tradd
    b\trmul\Bup{k-s+1}{s-1}{1}
    \nonumber
  \end{eqnarray}
  for $a=\delta_{k-s+1}$ and $b=\delta_{k-s-1}$.  By 
  \Cref{lemma:changemin}, we have
  $\Bup{k-s-1}{s+1}{2}\requal{s-1}\M{k-s+1}$ and thus the right-hand
  side of \cref{eq:equalityindstep} is true by the induction
  hypothesis. This concludes the induction to show
  \cref{eq:equalityinduction} and thus also the proof of the theorem.
\hspace*{\fill}\end{proof}

\paragraph*{Remark}
As a side-effect, the proof of \Cref{thm:equality} highlights
the fundamental significance of the \emph{critical path}
$P=(\delta_1,\delta_2,\ldots,\delta_k)$ in $CT_k$. After following
path $P$, the view of a node $v\in C_0$ ends up in the leaf-cluster
neighboring $C_0$ and sees $\delta_{i+1}$ neighbors.  Following the
same path, a node $v'\in C_1$ ends up in $C_0$ and sees
$\sum_{j=0}^{i}{\delta_j}-1$ neighbors. There is no way to match these
views. This inherent inequality is the underlying reason for the way
$G_k$ is defined: It must be ensured that the critical path is at
least $k$ hops long.

\subsection{Analysis} \label{sec:analysis}
In this subsection, we derive the lower bounds on the
approximation ratio of $k$-local MVC algorithms. Let $OPT$ be an
optimal solution for MVC and let $ALG$ be the solution computed by
any algorithm. The main observation is that adjacent nodes in the
clusters $C_0$ and $C_1$ have the same view and therefore, every
algorithm treats nodes in both of the two clusters the same way.
Consequently, $ALG$ contains a significant portion of the nodes of
$C_0$, whereas the optimal solution covers the edges between $C_0$
and $C_1$ entirely by nodes in $C_1$.

\begin{lemma}\label{lemma:labrand}
  Let $ALG$ be the solution of any distributed (randomized) vertex
  cover algorithm which runs for at most $k$ rounds. When applied to
  $G_k$ as constructed in \Cref{sec:construction} in the
  worst case (in expectation), $ALG$ contains at least half of the
  nodes of $C_0$.
\end{lemma}
\begin{proof}
  Let $v_0\in C_0$ and $v_1\in C_1$ be two arbitrary, adjacent nodes
  from $C_0$ and $C_1$. We first prove the lemma for deterministic
  algorithms.  The decision whether a given node $v$ enters the vertex
  cover depends solely on the topology $\mathcal{T}_{v,k}$ and the
  labeling $\mathcal{L}(\mathcal{T}_{v,k})$. Assume that the
  labeling of the graph is chosen uniformly at random. Further, let
  $p_0^\mathcal{A}$ and $p_1^\mathcal{A}$ denote the probabilities
  that $v_0$ and $v_1$, respectively, end up in the vertex cover when
  a deterministic algorithm $\mathcal{A}$ operates on the randomly
  chosen labeling. By \Cref{thm:equality}, $v_0$ and $v_1$ see
  the same topologies, that is,
  $\mathcal{T}_{v_0,k}=\mathcal{T}_{v_1,k}$.  With our choice of
  labels, $v_0$ and $v_1$ also see the same distribution on the
  labelings $\mathcal{L}(\mathcal{T}_{v_0,k})$ and
  $\mathcal{L}(\mathcal{T}_{v_1,k})$. Therefore it follows that
  $p_0^\mathcal{A}=p_1^\mathcal{A}$.

  We have chosen $v_0$ and $v_1$ such that they are neighbors in
  $G_k$. In order to obtain a feasible vertex cover, at least one of
  the two nodes has to be in it. This implies
  $p_0^\mathcal{A}+p_1^\mathcal{A}\ge1$ and therefore
  $p_0^\mathcal{A}=p_1^\mathcal{A}\ge1/2$. In other words, for all
  nodes in $C_0$, the probability to end up in the vertex cover is at
  least $1/2$. Thus, by the linearity of expectation, at least half of
  the nodes of $C_0$ are chosen by algorithm $\mathcal{A}$. Therefore,
  for every deterministic algorithm $\mathcal{A}$, there is at least
  one labeling for which at least half of the nodes of
  $C_0$ are in the vertex cover.\footnote{In fact, since at most
    $|C_0|$ such nodes can be in the vertex cover, for at least $1/3$
    of the labelings, the number exceeds $|C_0|/2$.}

  The argument for randomized algorithms is now straight-forward using
  Yao's minimax principle. The expected number of nodes chosen by a
  randomized algorithm cannot be smaller than the expected number of
  nodes chosen by an optimal deterministic algorithm for an
  arbitrarily chosen distribution on the labels.
  \hspace*{\fill}\end{proof}

\Cref{lemma:labrand} gives a lower bound on the number of
nodes chosen by any $k$-local MVC algorithm. In particular, we
have that $E[|ALG|]\ge |C_0|/2=n_0/2$. We do not know $OPT$, but
since the nodes of cluster $C_0$ are not necessary to obtain a
feasible vertex cover, the optimal solution is bounded by
$|OPT|\le n-n_0$. In the following, we define
\begin{equation}\label{eq:deltai}
  \delta_i := \delta^i\ \; , \; \forall i\in\{0,\dots,k+1\}
\end{equation}
for some value $\delta$. Hence, $\delta_0=1$ and for all $i\in\{0,\dots,k\}$, we
have $\delta_{i+1}/\delta_{i}=\delta$.

\begin{lemma}\label{lemma:nofnodes}
  If $\delta>k+1$, the number of nodes $n$ of $G_k$ is
  \[
  n\ \le\ n_0\left(1+\frac{k+1}{\delta-(k+1)}\right)
  \]
  and the largest degree $\Delta$ of $G_k$ is
  $\delta_{k+1}=\delta^{k+1}$.
\end{lemma}
\begin{proof}
  Consider a cluster $C$ of size $|C|$ on some level $\ell$ and some
  neighbor cluster $C'$ on level $\ell+1$. For some
  $i\in\{0,\dots,k\}$, all nodes in $C$ have $\delta_i$ neighbors in
  cluster $C'$ and all nodes in $C'$ have $\delta_{i+1}$ neighbors in
  cluster $C$. We therefore have
  $|C|/|C'|=\delta_{i+1}/\delta_i =\delta$ and for every $i$, $C$ can
  have at most $1$ such neighboring cluster on level $\ell+1$. The
  total number of nodes in level $\ell+1$ clusters that are neighbors
  of $C$ can therefore be bounded as $|C|\cdot (k+1)/\delta$.  Hence,
  the number of nodes decreases by at least a factor of
  $\delta/(k+1)$ on each level. For
  $\delta>k+1$, the total number of nodes can thus be
  bounded by
  \[
  \sum_{i=0}^\infty n_0\cdot\left(\frac{k+1}{\delta}\right)^i =
  n_0\cdot\frac{\delta}{\delta-(k+1)}.
  \]

  For determining the largest degree of $G_k$, observe that if in some
  cluster $C$, each node has $\delta_{k+1}$ neighbors in a neighboring
  cluster $C'$, $C'$ is the only neighboring cluster of $C$. Further,
  for each cluster $C$ and each $i\in\set{0,\dots,k+1}$, there is at
  most one cluster $C'$ such that nodes in $C$ have $\delta_i$
  neighbors in $C'$. The largest degree $\Delta$ of $G_K$ can
  therefore be computed as
  \[
  \Delta = \max\set{\delta_{k+1}, \sum_{i=0}^k \delta_i}.
  \]
  Because for each $i$, $\delta_{i+1}/\delta_i= \delta>2$,
  we have $\sum_{i=0}^k\delta_i < 2\delta_k =\delta_{k+1}$.
\hspace*{\fill}\end{proof}

It remains to determine the relationship between $\delta$ and $n_0$
such that $G_k$ can be realized as described in \Cref{sec:construction}. There, the construction of $G_k$ with large
girth is based on a smaller instance $G_k'$ where girth does not
matter. Using \cref{eq:deltai} (i.e.,
$\delta_i:=\delta^i$), we can now tie up this loose end
and describe how to obtain $G_k'$. Let $C_i$ and $C_j$ be two adjacent
clusters with $\ell(C_a,C_b)=(\delta_i,\delta_{i+1})$. We require that
$|C_a|/|C_b|=\delta_{i+1}/\delta_i=\delta$. Hence, $C_i$ and $C_j$ can simply
be connected by as many complete bipartite graphs
$K_{\delta_i,\delta_{i+1}}$ as necessary.

To compute the necessary cluster sizes to do this, let $c_\ell$ be the
size of the smallest cluster on level $\ell$. We have $c_0=n_0'$ and
$c_{\ell-1}/c_{\ell}=\delta_{k+1}/\delta_k=\delta$. Hence, the size
of the smallest cluster decreases by a factor $\delta$ when going
from some level to the next one. The smallest cluster $C_{\min}$ of
$G_k'$ is on level $k+1$. Because each node in the neighboring cluster
of $C_{\min}$ on level $k$ has $\delta_k$ neighbors in $C_{\min}$,
$C_{\min}$ needs to have size at least $\delta_k$. We thus choose
$C_{\min}$ of size $c_{k+1}=|C_{\min}|=\delta_k$. From
$c_{\ell-1}/c_{\ell}=\delta$, we thus get
\begin{equation*}
  n_0' = c_0 = c_{k+1}\cdot \delta^k =
  \delta_k\cdot \delta^k =
  \delta^{2k}.
\end{equation*}

If we assume that $\delta>2(k+1)$, we have $n'\le 2n_0'$, by 
\Cref{lemma:nofnodes}. Applying \Cref{lemma:largegirth} from
\Cref{sec:construction}, we can then construct $G_k$ with
girth $2k+1$ such that $n = O(n' \Delta^{(1+o(1))(2k+1)})$, where
$\Delta=\delta^{k+1}$ is the largest degree of $G_k'$ and $G_k$. Putting everything
together, we obtain
\begin{equation}\label{eq:nofnodesGk}
  n\ =\ O\left( n_0'\Delta^{(2k+1)(1+o(1))}\right)\
  \ =\ O\left(\delta^{4k^2(1+o(1))}\right).
\end{equation}
\begin{theorem}\label{thm:mvc-lowerbound}
  For every integer $k>0$,
  there are graphs $G$, such that in $k$ communication rounds in the
  LOCAL model, every distributed algorithm for the minimum vertex
  cover problem on $G$ has approximation ratios at least
  \[
  \Omega\left(\frac{n^{\frac{1-o(1)}{4k^2}}}{k}\right)\;\ \text{and}\quad
  \Omega\left(\frac{\Delta^{\frac{1}{k+1}}}{k}\right),
  \]
  where $n$ and $\Delta$ denote
  the number of nodes and the highest degree in $G$, respectively.
\end{theorem}
\begin{proof}
  We have seen that the size of an optimal vertex cover is
    at most $n-n_0$. For $\delta\geq 2(k+1)$, based on
  \Cref{lemma:labrand,lemma:nofnodes}, the approximation ratio of any
  $k$-round algorithm is therefore at least
  $\Omega(\delta/k)$. We thus need to bound $\delta$ as a
  function of the number of nodes $n$ and the largest degree $\Delta$
  of $G_k$.

  Assuming $\delta\geq2(k+1)$, by 
    \Cref{lemma:nofnodes}, we have $\Delta=\delta^{k+1}$ and
    $n=O(\delta^{4k^2(1+o(1))})$ and we thus have to choose $k$ such
    that $\Delta\leq (2(k+1))^{k+1}$ and
    $n = O\big((2(k+1))^{4k^2(1+o(1))}\big)$. If we choose $k$ such
    that $\Delta=\Theta((2(k+1))^{k+1})$ or
    $n=\Theta\big((2(k+1))^{4k^2(1+o(1))}\big)$, both claimed lower
    bounds simplify to $\Omega(1)$ and they thus trivially hold for
    such $k$ and also for larger $k$. If $k$ is chosen such that
    $\Delta\leq (2(k+1))^{k+1}$ and
    $n = O\big((2(k+1))^{4k^2(1+o(1))}\big)$, the lower bounds follow
    directly because $\delta=\Delta^{1/(k+1)}$ and
    $\delta=n^{(1-o(1))/4k^2}$ and because in this case, the
    (expected) approximation ratio of any (possibly randomized)
    $k$-round algorithm is at least $\Omega(\delta/k)$.
\hspace*{\fill}\end{proof}

\begin{theorem}\label{thm:log}
  In order to obtain a constant or polylogarithmic approximation ratio, even in the LOCAL
  model, every distributed algorithm for the MVC problem requires at
  least $\Omega\left(\sqrt{\log{n}/\log\log n}\right)$ and
  $\Omega\left(\log\Delta/\log\log\Delta\right)$ communication rounds.
\end{theorem}
\begin{proof}
  Follows directly from \Cref{thm:mvc-lowerbound}.
\hspace*{\fill}\end{proof}

\paragraph*{Remark} Note that the lower bounds of 
\Cref{thm:mvc-lowerbound,thm:log} hold even in the LOCAL
model (i.e., even if message size and local computations are not
bounded). Further, both lower bounds also hold even if the identifiers
of the nodes are $\set{1,\dots,n}$ and even if all nodes know the
exact topology of the network graph (i.e., in particular, every node
knows the exact values of $\Delta$ and $n$).


%% file: reductions.tex
\newcommand{\subtr}{\mathsf{sub}}

\section{Locality-Preserving Reductions}\label{sec:reductions}

Using the MVC lower bound, we can now derive lower bounds for
several of the other classical graph problems defined in
\Cref{sec:problems}. Interestingly, the \emph{hardness of
distributed approximation} lower bound on the MVC problem also gives
raise to local computability lower bounds for two of the most
fundamental exact problems in distributed computing: MIS and MM.

Specifically, we use the notion of \emph{locality preserving
reductions} to show that a number of other problems can be reduce to
MVC with regard to their local computability/approximability. This
implies that, like MVC, these problems fall into the polylog-local
class of problems. \Cref{fig:red} shows the hierarchy of locality preserving reductions derived in this section.

\begin{figure}[t]
  \begin{center}
    \includegraphics[width=0.95\columnwidth]{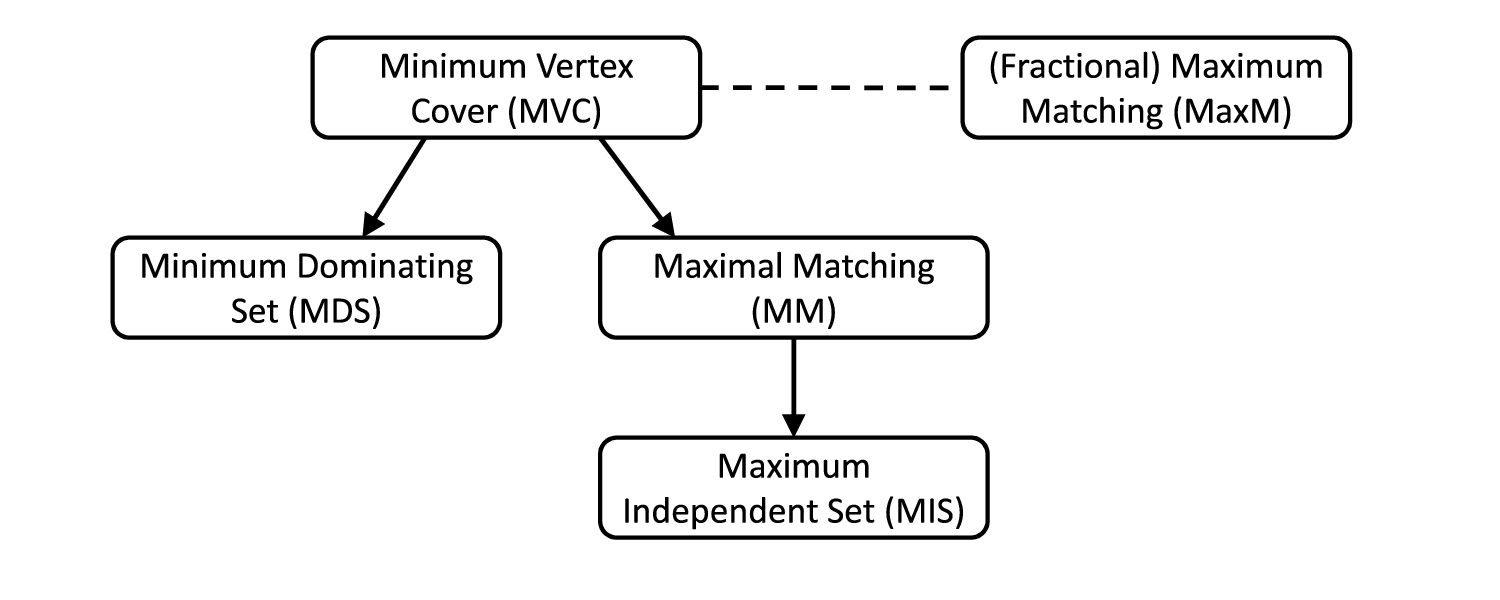}
    \caption{Locality Preserving Reductions. The dotted line between
      MVC and MaxM implies that we do not know of a direct locality
      preserving reduction between the covering and packing problem,
      but the respective lower bound constructions are based on a common cluster tree structure.  }\label{fig:red}
  \end{center}
\end{figure}

\subsection{Lower Bounds for Minimum Dominating Set}

In a non-distributed setting, MDS in equivalent to the general
minimum set cover problem, whereas MVC is a special case of set
cover which can be approximated much better. It is therefore not
surprising that also in a distributed environment, MDS is harder
than MVC. In the following, we formalize this intuition giving a
locality-preserving reduction from MVC to MDS.
\begin{theorem}\label{thm:mds-lowerbound}
  For every integer $k>0$, there are graphs $G$, such that in $k$
  communication rounds in the LOCAL model, every (possibly randomized)
  distributed algorithm for the minimum dominating set problem on $G$
  has approximation ratios at least
    \[
  \Omega\left(\frac{n^{\frac{1-o(1)}{4k^2}}}{k}\right)\quad \text{and}\quad
  \Omega\left(\frac{\Delta^{\frac{1}{k+1}}}{k}\right),
  \]
  where $n$ and $\Delta$ denote
  the number of nodes and the highest degree in $G$, respectively.
\end{theorem}
\begin{proof}
  To obtain a lower bound for MDS, we consider the line graph $L(G_k)$
  of $G_k$. The nodes of a line graph $L(G)$ of $G$ are the edges of
  $G$. Two nodes in $L(G)$ are connected by an edge whenever the two
  corresponding edges in $G$ are incident to the same node. Assuming
  that initially each node knows all its incident edges, a $k$-round
  computation on the line graph of $G$ can be simulated in $k$ round
  on $G$, i.e., in particular $G_k$ and $L(G_k)$ have the same
  locality properties.

  A dominating set of the line graph of a graph $G=(V,E)$ is a subset
  $E'\subseteq E$ of the edges such that for every $\set{u,v}\in E$,
  there is an edge $\set{u',v'}\in E'$ such that
  $\set{u,v}\cap\set{u',v'}\neq\emptyset$. Hence, for every edge
  dominating set $E'$, the node set $S=\bigcup_{\set{u,v}\in
    E'}\set{u,v}$ is a vertex cover of $G$ of size $|S|\leq 2|E'|$. In
  the other direction, given a vertex cover $S$ of $G$, we obtain a
  dominating set of $E'$ of $L(G)$ of the same size $|E'|=|S|$ simply
  by adding some edge $\set{u,v}$ to $E'$ for every node $u\in
  S$. Therefore, up to a factor of at most $2$ in the approximation
  ratio, the two problems are equivalent and thus the claim of the
  theorem follows.
\hspace*{\fill}\end{proof}

\paragraph*{Remark}
Using the same locality-preserving reduction as from MVC to MDS, it
can also be shown that solving the fractional version of MDS is at
least as hard as the fractional version of MVC. Since
\Cref{thm:mvc-lowerbound} also holds for fractional MVC (the
integrality gap of MVC is at most $2$),
\Cref{thm:mds-lowerbound} and
\Cref{cor:mds-lowerbound} can equally be stated for
fractional MDS, that is, for the standard linear programming
relaxation of MDS.

\begin{corollary}\label{cor:mds-lowerbound}
  In order to obtain a constant or polylogarithmic approximation ratio for minimum
  dominating set or fractional minimum dominating set, there are
  graphs on which even in the LOCAL model, every distributed algorithm
  requires time
  \[
  \Omega\left(\sqrt{\frac{\log n}{\log\log n}}\right)
  \;\ \text{and}\quad
  \Omega\left(\frac{\log\Delta}{\log\log\Delta}\right).
  \]
\end{corollary}
\begin{proof}
  The corollary follows directly from \Cref{thm:mds-lowerbound}.
\hspace*{\fill}\end{proof}

\paragraph*{Remark}
The MDS problem on the line graph of $G$ is also known as the minimum
edge dominating set problem of $G$ (an edge dominating set is a set of
edges that 'covers' all edges). Hence, the above reduction shows that
also the minimum edge dominating set problem is hard to approximate
locally.

\subsection{Lower Bounds for Maximum Matching}

While MVC and MDS are standard covering problems, the lower bound can
also be extended to \emph{packing problems}. Unfortunately, we are not
aware of a simple locality-preserving reduction from MVC to a packing
problem, but we can derive the result by appropriately adjusting the
cluster graph from \Cref{sec:construction}. In fact, we prove
the result for the \emph{fractional maximum matching problem} in which
edges may be selected fractionally, and the sum of these fractional
values incident at a single node must not exceed $1$. Let $E(v)$
denotes the set of edges incident to node~$v$.

The basic idea of the lower bound follows along the lines of the MVC
lower bound in \Cref{sec:lower}. The view of an edge $e=(u,v)$
can be defined as the union of its incident nodes' views along with a
specification, which one edge $e$ is in both node views. In 
\Cref{lemma:edgeview}, we first show that if the graph has large girth
so that all node views are trees, the topology of edge view is
uniquely defined by the topologies of the views of the two nodes. The
common edge does not have to be specified explicitly in this case. In
other words, in graphs with large girth, two edges $(u,v)$ and
$(u',v')$ have the same view if $\mathcal{V}_{u,k}=\mathcal{V}_{u',k}$
and $\mathcal{V}_{v,k}=\mathcal{V}_{v',k}$.

The idea is to construct a graph $H_k$ which contains a large set
$E'\subset E$ of edges with equal view up to distance $k$. This
implies that, in expectation, the fractional values $y_e$ assigned to
the edges in $E'$ must be equal. $H_k$ is constructed in such a way,
that there are edges in $E'$ that are incident to many other edges in
$E'$. Further, the edges in $E'$ also contain a large matching of
$H_k$ and all large matchings of $H_k$ predominantly consist of edges
in $E'$. As every distributed $k$-local algorithm assigns equal
fractional values $y_e$ to all edges in $E'$ in expectation, in order
to keep the feasibility at the nodes incident to many edges in $E'$,
this fractional value must be rather small. Together with the fact
that all large matchings have to consist of a large number of edges
from $E'$, this will lead to the sub-optimality captured in
\Cref{thm:matching-lowerbound}.

\begin{figure}[t]
  \begin{center}
    \includegraphics[width=0.7\columnwidth]{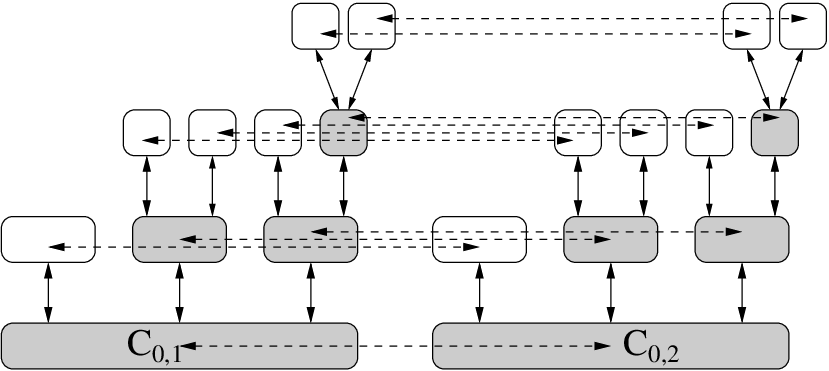}
    \caption{The structure of lower-bound graph $H_k$.}\label{fig:matching}
  \end{center}
\end{figure}

The construction of $H_k$ uses the lower-bound graph $G_k$ of the MVC
lower bound. Essentially, we take two identical copies $G_{k,1}$ and
$G_{k,2}$ of the MVC lower bound graph defined in
\Cref{sec:lower} and we connect $G_{k,1}$ and $G_{k,2}$ to each
other by using a perfect matching. Formally, in order to obtain a
graph $H_k$ with large girth, we start with two copies of $G_{k,1}'$
and $G_{k,2}'$ with low girth (and fewer nodes). We then obtain a
graph $H_k'$ by adding an edge between each node in $G_{k,1}'$ and its
corresponding edge in $G_{k,2}'$ (i.e., the edges connecting
$G_{k,1}'$ and $G_{k,2}'$ form a perfect matching of $H_k'$). Clearly,
the graph $H_k'$ has cycles of length $4$ (formed by two corresponding
edges in $G_{k,1}'$ and $G_{k,2}'$ and by two of the edges connecting
$G_{k,1}'$ and $G_{k,2}'$). In order to obtain a copy of $H_k'$ with
large girth, we compute a lift $H_k$ of $H_k'$ by applying the same
construction as in \Cref{lemma:largegirth} from 
\Cref{sec:construction}. As a result, we obtain a graph $H_k$ with
girth at least $2k+2$ such that $H_k$ consists of two (large-girth)
copies $G_{k,1}$ and $G_{k,2}$ of the MVC lower bound graph, where
corresponding nodes in $G_{k,1}$ and $G_{k,2}$ are connected by an
edge. Hence, also in $H_k$, the edges connecting $G_{k,1}$ and
$G_{k,2}$ for a perfect matching of the graph. In the following, we
use $C_{i,1}$ and $C_{i,2}$ to denote the copies of cluster $C_i$ in
graphs $G_{k,1}$ and $G_{k,2}$. Furthermore, we use the abbreviations
$S_0:=C_0\cup C'_0$ and $S_1:=C_1\cup C'_1$.  The structure of $H_k$
is illustrated in \Cref{fig:matching}. We start by showing that
all edges between clusters $C_{0,1}$, $C_{1,1}$, $C_{0,2}$, and
$C_{1,2}$ have the same view up to distance $k$. As stated above, we
first show that in trees (and thus in graphs of sufficiently large
girth), the views of two edges are identical if all four nodes have
the same view.

Recall that we use $\mathcal{T}_{v,k}$ to denote the topology of the
$k$-hop view of node $v$ in a given graph $G$. If $G$ has girth at
least $2k+1$, $\mathcal{T}_{v,k}$ is a tree of depth at most $k$ (it
is the tree induced by all nodes at distance at most $k$ from
$v$). For an edge $e=\set{u,v}$, we define the $k$-hop topology
$\mathcal{T}_{e,k}$ as the ``union'' of $\mathcal{T}_{u,k}$ and
$\mathcal{T}_{v,k}$. Hence if $G$ has girth at least $2k+2$,
$\mathcal{T}_{e,k}$ is the tree induced by all nodes at distance at
most $k$ from $u$ or $v$.

\begin{lemma}\label{lemma:edgeview}
  Let $G=(V,E)$ be a graph and let $u$, $v$, $x$, and $y$ be $4$ nodes
  such that $e=\set{u,v}\in E$, $e'=\set{x,y}\in E$, and for some
  $k\geq 1$,
  $\mathcal{T}_{u,k}=\mathcal{T}_{v,k}=\mathcal{T}_{x,k}=\mathcal{T}_{y,k}$. If
  the girth of $G$ is at least $2k+2$, the $k$-hop topologies of $e$
  and $e'$ are identical, i.e.,
  $\mathcal{T}_{e,k}=\mathcal{T}_{e',k}$.
\end{lemma}
\begin{proof}
  First note that from the girth assumption, it follows that all the
  considered $k$-hop topologies $\mathcal{T}_{u,k}$,
  $\mathcal{T}_{v,k}$, $\mathcal{T}_{x,k}$, $\mathcal{T}_{y,k}$,
  $\mathcal{T}_{e,k}$, and $\mathcal{T}_{e',k}$ are trees. Further,
  the assumption that the $k$-hop topology of the four nodes $u$, $v$,
  $x$, and $y$ are identical implies that
  $\mathcal{T}_{u,i}=\mathcal{T}_{v,i}=\mathcal{T}_{x,i}=\mathcal{T}_{y,i}$
  for all $i\in\set{0,\dots,k}$.

  We show that $\mathcal{T}_{e,i}=\mathcal{T}_{e',i}$ for all
  $i\in\set{0,\dots,k}$ by induction on $i$. To show this, for an
  unlabeled tree $T$ and a node $u\in T$ and a neighbor $v\in T$, let
  $\subtr_v(T,u)$ be the unlabeled subtree of node $u$ rooted at node
  $v$. Further, let $\subtr(T,u)$ be the multiset containing
  $\subtr_v(T,u)$ for all neighbors $v$ of $u$ in $T$. Note that if
  the topology $\mathcal{T}_{u,k}$ of the $k$-hop view of a node $u$
  is a tree, $\mathcal{T}_{u,k}$ is uniquely described by
  $\subtr(\mathcal{T}_{u,k},u)$. To prove the lemma, we show by
  induction on $i$ that for all $i\in\set{1,\dots,k}$ and for the
  nodes $u$, $v$, $x$, and $y$ of $G$, we have
  \begin{equation}\label{eq:subtreeequivalence}
    \subtr_v(\mathcal{T}_{u,i},u) =
    \subtr_u(\mathcal{T}_{v,i},v) =
    \subtr_y(\mathcal{T}_{x,i},x) =
    \subtr_x(\mathcal{T}_{y,i},y).
  \end{equation}
  Note that because the nodes $u$, $v$, and $x$, and $y$ are assumed
  to have identical $k$-hop views, for all $i\in\set{1,\dots,k}$ we
  also clearly have
  \begin{equation}\label{eq:samesubtrees}
    \subtr(\mathcal{T}_{u,i},u) =
    \subtr(\mathcal{T}_{v,i},v) =
    \subtr(\mathcal{T}_{x,i},x) =
    \subtr(\mathcal{T}_{y,i},y).
  \end{equation}
  \Cref{eq:subtreeequivalence} holds for $i=1$ because all
  the four subtrees are single nodes. For example, for
  $\mathcal{T}{u,1}$ is a star with center $u$ and the subtree rooted
  at neighbor $v$ is the node $v$ itself. For the induction step, let
  us assume that \cref{eq:subtreeequivalence} holds for $i=i_0<k$ and
  we want to show that it also holds for $i=i_0+1$. Let us first
  construct $\subtr_v(\mathcal{T}_{u,i_0+1},u)$. The subtree of $u$
  rooted at $v$ in $\mathcal{T}_{u,i_0+1}$ is uniquely determined by
  the $i_0$-hop view of node $v$. It consists of root node $v$ with
  subtrees $\subtr(\mathcal{T}_{v,i_0},v)\setminus
  \subtr_u(\mathcal{T}_{v,i_0},v)$. By the assumption that 
  \cref{eq:subtreeequivalence,eq:samesubtrees} hold for
  $i=i_0$, we can then conclude that
  $\subtr_v(\mathcal{T}_{u,i_0+1},u)=\subtr_u(\mathcal{T}_{v,i_0+1},v)$
  and by symmetry also that \cref{eq:subtreeequivalence} holds
  for $i=i_0+1$. This proves the claim of the lemma.
  \hspace*{\fill}%
\end{proof}

By the construction of $H_k$ and the structural properties proven
in \Cref{thm:equality}, the following lemma now follows
in a straightforward way.
\begin{lemma}\label{lm:matching-equal}
  Let $\set{u,v}$ and $\set{u',v'}$ be two edges of $H_k$ such that
  $u$, $v$, $u'$, and $v'$ are four nodes in $S_0\cup S_1$. Then, the
  two edges see the same topology up to distance $k$.
\end{lemma}
\begin{proof}
  Let $E'$ be the set of edges connecting $G_{k,1}$ and $G_{k,2}$.  As
  the girth of $H_k$ is at least $2k+2$, the $k$-hop views of all four
  nodes and also the $k$-hop views of the two edges are trees.  By
 \Cref{thm:equality}, when removing the edges in $E'$, all
  four nodes have the same $k$-hop view. As each node in $w\in S_0\cup
  S_1$ is incident to exactly one edge in $E'$, connecting $w$ to a
  node $w'\in S_0\cup S_1$, also after adding the edges in $E'$, all
  four nodes have the same $k$-hop view. By 
  \Cref{lemma:edgeview}, also the two edges have the same $k$-hop view
  and therefore the lemma follows.
\end{proof}

\Cref{lm:matching-equal} implies that no distributed $k$-local
algorithm can distinguish between edges connecting two nodes in
$S_0\cup S_1$. In particular, this means that edges between $C_{0,i}$
and $C_{0,i}$ (for $i\in\set{1,2}$) cannot be distinguished from edges
between $C_{0,1}$ and $C_{0,2}$. In the sequel, let $OPT$ be the value
of the optimal solution for fractional maximum matching and let $ALG$
be the value of the solution computed by any algorithm.
\begin{lemma}\label{lm:matching-labrand}
  When applied to $H_k$, any distributed, possibly randomized
  algorithm which runs for at most $k$ rounds computes, in
  expectation, a solution of at most
  $ALG \le\ |S_0|/(2\delta)+(|V|-|S_0|)$.
\end{lemma}
\begin{proof}
  First, consider deterministic algorithms. The decision of which
  value $y_e$ is assigned to edge $e=(v,v)$ depends only on the view
  the topologies $\mathcal{T}_{u,k}$ and $\mathcal{T}_{v,k}$ and the
  labelings $\mathcal{L}(\mathcal{T}_{u,k})$ and
  $\mathcal{L}(\mathcal{T}_{v,k})$, which $u$ and $v$ can collect
  during the $k$ communication rounds. Assume that the labeling of
  $H_k$ is chosen uniformly at random. In this case, the labeling
  $\mathcal{L}(\mathcal{T}_{u,k})$ for any node $u\in V$ is also
  chosen uniformly at random.

  All edges connecting nodes in $S_0$ and $S_1$ see the same
  topology. If the node's labels are distributed uniformly at random,
  it follows that the \emph{distribution of the views} (and therefore
  the distribution of the $y_e$) is the same for all edges connecting
  nodes in $S_0$ and $S_1$. We denote the random variables describing
  the distribution of the $y_e$ by $Y_e$. Every node $u\in S_1$ has
  $\delta_1=\delta$ neighbors in $S_0$. Therefore, for edges $e$
  between nodes in $S_0$ and $S_1$, it follows by linearity of
  expectation that $E[Y_e]\le 1/\delta$ because otherwise, there
  exists at least one labeling for which the computed solution is not
  feasible. On the other hand, consider an edge $e'$ having both
  end-points in $S_0$. By \Cref{lm:matching-equal}, these edges
  have the same view as edges $e$ between $S_0$ and $S_1$. Hence, for
  $y_e'$ of $e'$, it must equally hold that $E[Y_e']\le
  1/\delta$. Because there are $|S_0|/2$ such edges, the expected
  total value contributed to the objective function by edges between
  two nodes in $S_0$ is at most $|S_0|/(2\delta)$.

  Next, consider all edges which do not connect two nodes in $S_0$.
  Every such edge has at least one end-point in $V\setminus S_0$. In
  order to obtain a feasible solution, the total value of all edges
  incident to a set of nodes $V'$, can be at most $|V'|=|V\setminus
  S_0|$. This can be seen by considering the dual problem, a kind of
  minimum vertex cover where some edges only have one incident node.
  Taking all nodes of $V'$ (assigning 1 to the respective variables)
  yields a feasible solution for this vertex cover problem. This
  concludes the proof for deterministic algorithms.

  For probabilistic algorithms, we can apply an identical argument
  based on Yao's minimax principle as in the MVC lower bound
  (cf.~\Cref{lemma:labrand}).
\hspace*{\fill}\end{proof}

\Cref{lm:matching-labrand} yields an upper bound on the
objective value achieved by any $k$-local fractional maximum
matching algorithm. On the other hand, it is clear that choosing
all edges connecting corresponding nodes of $G_k$ and $G'_k$ is
feasible and hence, $OPT\geq n/2 \geq |S_0|/2$. Let $\alpha$
denote the approximation ratio achieved by any $k$-local
distributed algorithm, and assume---as in the MVC proof---that
$k+1\leq \delta/2$. Using the relationship between $n$, $|S_0|$,
$\delta$, and $k$ proven in \Cref{lemma:nofnodes} and
combining it with the bound on $ALG$ gives raise to the following
theorem.

\begin{theorem}\label{thm:matching-lowerbound}
  For every integer $k>0$, there are graphs $G$, such that in $k$
  communication rounds in the LOCAL model, every (possibly randomized)
  distributed algorithm for the (fractional) maximum matching problem
  on $G$ has approximation ratios at least
  \[
  \Omega\left(\frac{n^{\frac{1-o(1)}{4k^2}}}{k}\right)\quad \text{and}\quad
  \Omega\left(\frac{\Delta^{\frac{1}{k+1}}}{k}\right),
  \]
  where $n$ and $\Delta$ denote the number
  of nodes and the highest degree in $G$, respectively.
\end{theorem}

\begin{proof}
  By \Cref{lm:matching-labrand,lemma:nofnodes}, on $H_k$, the approximation ratio of any,
  possibly randomized, (fractional) maximum matching algorithm is
  $\Omega(\delta)$. Because asymptotically, the relations between
  $\delta$ and the largest degree $\Delta$ and the number of nodes $n$
  is the same in the MVC lower bound graph $G_k$ and in $H_k$, the
  lower bounds follow in the same way as the lower bounds in 
  \Cref{thm:mvc-lowerbound}.
\hspace*{\fill}\end{proof}

\begin{corollary}\label{cor:matching-log}
  In order to obtain a constant or polylogarithmic
  approximation ratio, even in the LOCAL model, every distributed
  algorithm for the (fractional) maximum matching problem requires at
  least
  \[
  \Omega\left(\sqrt{\log{n}/\log\log n}\right) \;\
  \text{and}\quad \Omega\left(\log\Delta/\log\log \Delta\right)
  \]
  communication rounds.
\end{corollary}

\subsection{Lower Bounds for Maximal Matching}

A maximal matching M of a graph $G$ is a maximal set of edges
which do not share common end-points. Hence, a maximal matching is
a set of non-adjacent edges $M$ of $G$ such that all edges in
$E(G)\setminus M$ have a common end-point with an edge in M. The
best known lower bound for the distributed computation of a
maximal matching is $\Omega(\log^*\!n)$ which holds for rings
\cite{linial92}.

\begin{theorem}\label{thm:mm-lowerbound}
  There are graphs $G$ on which every distributed, possibly randomized
  algorithm in expectation requires time
  \[
  \Omega\left(\sqrt{\log{n}/\log\log n}\right)
  \;\ \text{and}\quad
  \Omega\left(\log\Delta/\log\log\Delta\right)
  \]
  to compute a maximal matching. This bound holds even in the LOCAL
  model, i.e. even if message size is unlimited and nodes have unique
  identifiers.
\end{theorem}
\begin{proof}
  It is well known that the set of all end-points of the edges of a
  maximal matching form a 2-approximation for MVC. This simple
  2-approximation algorithm is commonly attributed to Gavril and
  Yannakakis. For deterministic algorithms, the lower bound for the
  construction of a maximal matching in
  \Cref{thm:mm-lowerbound} therefore directly follows from
  \Cref{thm:log}.

  Generalizing this result to randomized algorithms, however, still
  requires some work. The problem is that
  \Cref{thm:mvc-lowerbound} lower bounds the achievable
  approximation ratio by distributed algorithms whose time complexity
  is exactly $k$. That is, it does not provide a lower bound for
  randomized algorithms whose time complexity is at most $k$ in
  expectation or with a certain probability. As stated in the theorem,
  however, we consider distributed algorithms that always compute a
  feasible solution, i.e., only the time complexity depends on
  randomness. In other words, \Cref{thm:mvc-lowerbound} yields
  a bound on Monte Carlo type algorithms, whereas in the case of
  maximal matching, we are primarily interested in Las Vegas type
  algorithms.

  In order to generalize the theorem to randomized algorithms, we give
  a transformation from an arbitrary distributed maximal matching
  algorithm $\mathcal{A}_M$ with expected time complexity $T$ into a
  distributed vertex cover algorithm $\mathcal{A}_{\mathrm{VC}}$ with fixed
  time complexity $2T+1$ and expected approximation ratio
  $11$.

  We first define an algorithm $\mathcal{A}'_{\mathrm{VC}}$. In a first phase,
  $\mathcal{A}'_{\mathrm{VC}}$ simulates $\mathcal{A}_M$ for exactly $2T$
  rounds. Let $E_M\subseteq E$ be the set of edges selected after
  these rounds. In the second phase, every node $v$ checks whether it
  has at most one incident edge in $E_{\mathrm{VC}}$. If a node has more than
  one incident edge in $E_{\mathrm{VC}}$, it removes all these edges from
  $E_{\mathrm{VC}}$. Hence, $E_{\mathrm{VC}}$ forms a feasible matching, although not
  necessarily a maximal one.

  It follows from Markov's inequality that when running
  $\mathcal{A}_M$ for $2T$ rounds, the probability for obtaining a
  feasible maximal matching is at least $1/2$. Therefore, algorithm
  $\mathcal{A}'_{\mathrm{VC}}$ outputs a matching that is maximal with
  probability at least $1/2$. Let $V_{\mathrm{VC}}\subseteq V$ denote the set
  of all nodes incident to an edge in $E_{\mathrm{VC}}$. If $E_{\mathrm{VC}}$ is a
  maximal matching, $V_{\mathrm{VC}}$ is a feasible vertex cover (with
  probability at least $1/2$). In any case, the construction of
  $\mathcal{A}'_{\mathrm{VC}}$ guarantees that $|V_{\mathrm{VC}}|$ is at most twice the
  size of an optimal vertex cover.

  Algorithm $\mathcal{A}_{\mathrm{VC}}$ executes $c\cdot\ln\Delta$ independent runs
  of $\mathcal{A}'_{\mathrm{VC}}$ in parallel for a sufficiently large constant
  $c$. Let $V_{\mathrm{VC},i}$ be the node set $V_{\mathrm{VC}}$ constructed by the
  $i^{\mathit{th}}$ of the $c\cdot\ln\Delta$ runs of Algorithm
  $\mathcal{A}_{\mathrm{VC}}$. For each node $u\in V$, we define
  \[
  x_u := 6\cdot\frac{\left|\set{i:u\in V_{\mathrm{VC},i}}\right|}{c\cdot\ln\Delta}.
  \]
  Algorithm $\mathcal{A}_{\mathrm{VC}}$ computes a vertex cover $S$ as
  follows. All nodes with $x_u\geq1$ join the initial set $S$. In one
  additional round, nodes that have an uncovered edge also join $S$ to
  guarantee that $S$ is a vertex cover.

  Let $\mathit{OPT}_{\mathrm{VC}}$ be the size of an optimal vertex
  cover. Because for each $i$, $|V_{\mathrm{VC},i}|\leq 2\mathit{OPT}_{\mathrm{VC}}$, we
  get $\sum_{u\in V} x_u \leq 12\cdot \mathit{OPT}_{\mathrm{VC}}$. For every
  edge $\set{u,v}$, in each run of $\mathcal{A}'_{\mathrm{VC}}$ that ends with
  a vertex cover, the set $\set{i:u\in V_{\mathrm{VC},i}}$ contains at least
  one of the two nodes $\set{u,v}$. Hence, if at least $1/3$ of the
  runs of $\mathcal{A}'_{\mathrm{VC}}$ produces a vertex cover, we have
  $x_u+x_v\geq2$ for every edge $\set{u,v}$. Thus, in this case,
  taking all nodes $u$ for which $x_u\geq 1$ gives a valid vertex
  cover of size at most $\sum_{u\in V} x_u$. Let $X$ be the number of
  runs of $\mathcal{A}'_{\mathrm{VC}}$ that result in a vertex cover. Because
  the runs are independent and since each of them gives a vertex cover
  with probability at least $1/2$, we can bound the number of
  successful runs using a Chernoff bound:
  \[
  \Pr\left[X < \frac{c\ln\Delta}{3}\right] =
  \Pr\left[X <
    \left(1-\frac{1}{3}\right)\cdot\frac{c\ln\Delta}{2}\right] \leq
    e^{-\frac{c}{36}\ln\Delta} = \frac{1}{\Delta^{c/36}}.
  \]
  For $c\geq36$, the probability that the nodes $u$ with $x_u\geq1$ do
  not form a vertex cover is at most $1/\Delta$. Thus, with
  probability at least $1-1/\Delta$, the algorithm computes a vertex
  cover of size at most $10\mathit{OPT}_{\mathrm{VC}}$. With probability at
  most $1/\Delta$, the vertex cover has size at most
  $n\leq\Delta\mathit{OPT}_{\mathrm{VC}}$. The expected size of the computed
  vertex cover therefore is at most $11\mathit{OPT}_{\mathrm{VC}}$. The theorem
  now follows from \Cref{thm:log}.
\hspace*{\fill}\end{proof}

\subsection{Lower Bounds for Maximal Independent Set (MIS)}

As in the case of a maximal matching, the best currently known lower
bound on the distributed complexity of an MIS has been Linial's
$\Omega(\log^*\!n)$ lower bound. Using a locality-preserving
reduction from MM to MIS, we can strengthen this lower bound on
general graphs as formalized in the following theorem.
\begin{theorem}\label{thm:mmmis}
  There are graphs $G$ on which every distributed, possibly randomized
  algorithm in expectation requires time
  \[
  \Omega\left(\sqrt{\log{n}/\log\log n}\right)
  \;\ \text{and}\quad
  \Omega\left(\log\Delta/\log\log\Delta\right)
  \]
  to compute a maximal independent set (MIS). This bound holds even in
  the LOCAL model, i.e., even if message size is unlimited and nodes
  have unique identifiers.
\end{theorem}
\begin{proof}
  For the MIS problem, we again consider the line graph $L(G_k)$ of
  $G_k$, i.e., the graph induced by the edges of $G_k$.
  The MM problem on a graph $G$ is equivalent to the MIS problem on
  $L(G)$. Further, if the real network graph is $G$, $k$ communication
  rounds on $L(G)$ can be simulated in $k+\Oh(1)$ communication
  rounds on $G$.  Therefore, the times $t$ to compute an MIS on
  $L(G_k)$ and $t'$ to compute an MM on $G_k$ can only differ by a
  constant, $t\ge t'-O(1)$. Let $n'$ and $\Delta'$ denote the number
  of nodes and the maximum degree of $G_k$, respectively. The number
  of nodes $n$ of $L(G_k)$ is less than $n'^2/2$, the maximum degree
  $\Delta$ of $G_k$ is less than $2\Delta'$. Because $n'$ only appears
  as $\log n'$, the power of $2$ does not hurt and the theorem holds
  ($\log n = \Theta(\log n')$).
\hspace*{\fill}\end{proof}

\include{CDSlowerbound}


%% file: CDSlowerbound.tex
\subsection{Connected Dominating Set Lower Bound}
\label{sec:cdslower}

In this section, we extend our lower bound to the minimum connected
dominating set problem (MCDS). First first start with a simple
technical lemma that relates the respective sizes of an optimal
dominating set and an optimal connected dominating set in a graph.
\begin{lemma}\label{lemma:distcds}
  Let $G=(V,E)$ be a connected graph and let $DS_{OPT}$ and $CDS_{OPT}$ be the sizes
of optimal dominating and connected dominating sets of G. It holds
that $CDS_{OPT} < 3 \cdot DS_{OPT}$. Moreover, every dominating set
$D$ of $G$ can be turned into a connected dominating set
$D'\supseteqD$ of size $|D'| < 3|D|$.
\end{lemma}
\begin{proof}
Given $G$ and $D$, we define a graph $G_D = (V_D,E_D)$ as follows.
$V_D = D$ and there is an edge $(u,v)\in E_D$ between $u,v\in D$ if
and only if $d_G(u, v)\leq 3$. We first show that $G_D$ is
connected. For the sake of contradiction assume that $G_D$ is not
connected. Then there is a cut $(S,T)$ with $S\subseteq D, T = D
\setminus S$, and $S,T \neq \emptyset$ such that 
\begin{equation}\label{Eq:cds1}
\forall u \in S, \forall v \in T \; : \;d_G(u, v)\geq 4.
\end{equation}
Let $u\in S$ and $v\in T$ be such that
\begin{equation}\label{Eq:cds2}
d_G(u, v) = \min_{u\in S,v\in T}{(d_G(u, v))}.
\end{equation}
By \Cref{Eq:cds1}, there is a node $w\in V$ with
$d_G(u,w)\geq 2$ and $d_G(v,w)\geq 2$ on each shortest path
connecting $u$ and $v$. Because of \Cref{Eq:cds2}, we have
that
\begin{equation}\nonumber
\forall u \in S, \forall v \in T \; : \; d_G(u,w)\geq 2 \wedge
d_G(v,w)\geq 2.
\end{equation}
However this is a contradiction to the assumption that $D=S\cup T$
is a dominating set of $G$.

We can now construct a connected dominating set $D'$ as follows. We
first compute a spanning tree of $G_D$. For each edge $(u,v)$ of the
spanning tree, we add at most two nodes such that $u$ and $v$ become
connected. Because the number of edges of the spanning tree is
$|D|-1$, this results in a connected dominating set of size at most
$3|D|-2$.
\hspace*{\fill}\end{proof}

Using this lemma, we can now derive the lower bound on the local
approximability of the MCDS problem.

\begin{figure}[t]
  \centering
  \includegraphics[width=0.92\textwidth]{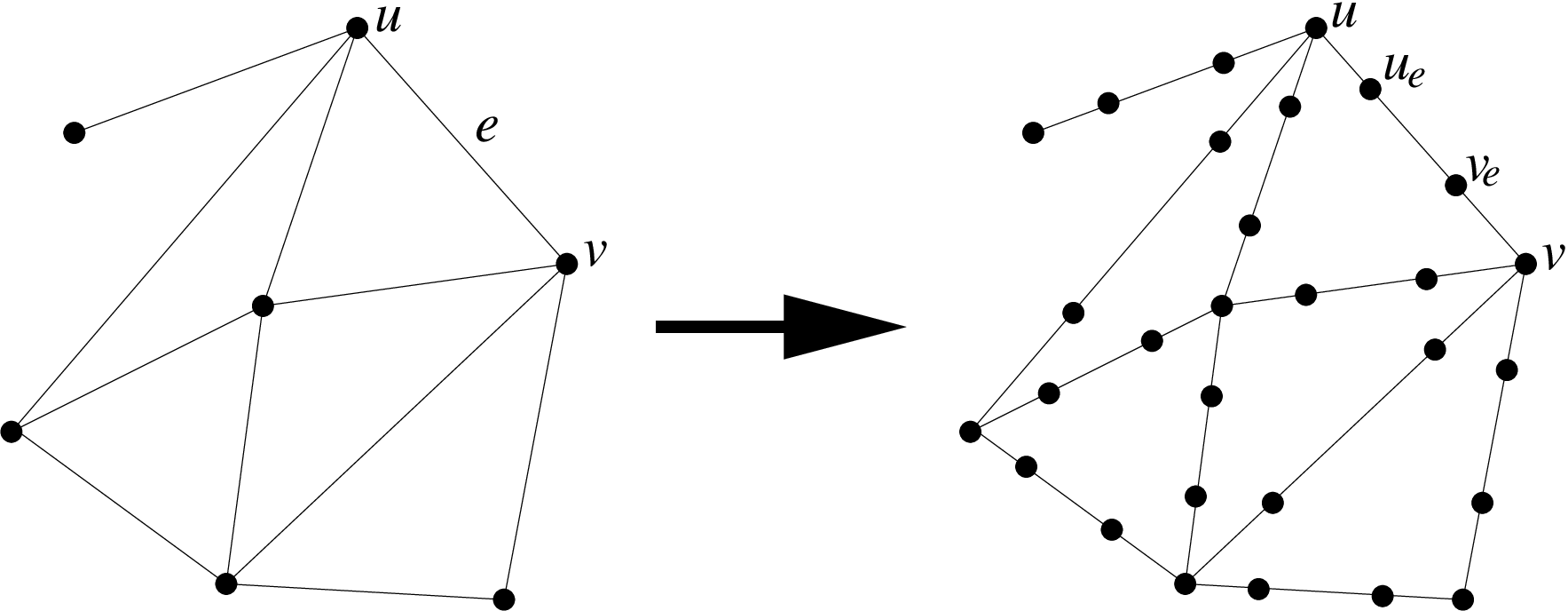}
  \caption{Graph transformation used for the distributed minimum connected
    dominating set lower bound}
  \label{fig:cdslowerbound}
\end{figure}

\begin{theorem}
  Consider a (possibly randomized) $k$-round algorithm for the MCDS
  problem. There are graphs for which every such algorithm computes a
  connected dominating set $S$ of size at least
  \[
  |S|\ \ge\ n^{\Omega(1/k)}\cdot\mathrm{CDS_{OPT}},
  \]
  where $\mathrm{CDS_{OPT}}$ denotes the size of an optimal connected
  dominating set.
\end{theorem}
\begin{proof}
  It follows from a well-known theorem (see e.g.~\cite{bollobas-book}) that there exist graphs $G=(V,E)$ with
  girth $g(G)\ge (2k+1)/3$ and number of edges $|E|=n^{1+\Omega(1/k)}$. From any such graph $G$, we construct a graph $G'$ as follows. For
  every edge $e=(u,v)\in E$, we generate additional nodes $u_e$ and
  $v_e$. In $G'$, there is an edge between $u$ and $u_e$, between
  $u_e$ and $v_e$, and between $v$ and $v_e$. Note that there is no
  edge between $u$ and $v$ anymore. The described transformation is
  illustrated in \Cref{fig:cdslowerbound}. We denote the set of
  all new nodes by $W$ and the number of nodes of $G'$ by $N=|V\cup
  W|$.

  By the definition of $G'$, the nodes in $V$ form a dominating set of
  $G'$. Hence, an optimal connected dominating set on $G'$ has size
  less than $3|V|$ by \Cref{lemma:distcds}. Note that the
  construction described in \Cref{lemma:distcds} actually
  computes a spanning tree $T$ on $G$ and adds all nodes $u_e,v_e\in
  W$ to the dominating set for which $(u,v)$ is an edge of $T$. To
  bound the number of nodes in the connected dominating set of a
  distributed algorithm, we have a closer look at the locality
  properties of $G'$. Because $g(G)\ge(2k+1)/3$, the girth of $G'$ is
  $g(G')\ge2k+1$. This means that in $k$ communication rounds it is
  not possible to detect a cycle of $G'$. Hence, no node can locally
  distinguish $G'$ from a tree. However, since on a tree all edges are
  needed to keep a connected graph, a $k$-round algorithm removing an
  edge from $G'$ cannot guarantee that the resulting topology remains
  connected. This means that the connected dominating set of every
  $k$-round algorithm must contain all nodes $V\cup W$ of $G'$. The
  approximation ratio of every distributed $k$-round MCDS algorithm on
  $G'$ is therefore bounded by
  \[
  \frac{|V\cup W|}{3|V|}\ =\ \frac{n^{1+\Omega(1/k)}}{3n}\ =\ \
  n^{\Omega(1/k)}\ =\
  N^{\left(1-\frac{1}{k+\Omega(1)}\right)\Omega(1/k)}\ =\
  N^{\Omega(1/k)}.
  \]
\hspace*{\fill}\end{proof}


%% file: upperbound.tex
\section{Local Computation: Upper Bounds}
\label{sec:upper}

This section is devoted to distributed algorithms with similar
time-approximation guarantees as given by the lower bounds in 
\Cref{sec:lower} for the problems introduced in 
\Cref{sec:problems}. In \Cref{sec:mvc-upper}. we start with a
simple algorithm that specifically targets the minimum vertex cover
problem and asymptotically achieves the trade-off given by the
$\Omega(\Delta^{\frac{1-\eps}{k+1}})$ lower bound in \Cref{thm:mvc-lowerbound}. We then describe a generic distributed
algorithm to approximate covering and packing linear programs in
\Cref{sec:lp-upper}. In \Cref{sec:rounding,sec:cds-upper}, we show how an LP solution can be turned into a
solution for vertex cover, dominating set, matching, or a related
problems by randomized rounding and how a dominating set can be
extended to a connected dominating set. Finally, we conclude this
section by providing a derandomization result for the distributed
solution of fractional problems and with a general discussion on the
role of randomization and fractional relaxations in the context of
local computations in \Cref{sec:randomization}.

\subsection{Distributed Vertex Cover Algorithm}
\label{sec:mvc-upper}

The MVC problem appears to be an ideal starting point for studying
distributed approximation algorithms. In particular, as long as we are
willing to pay a factor of $2$ in the approximation ratio, MVC does
not involve the aspect of \emph{symmetry breaking} which is so crucial
in more complex combinatorial problems. The fractional relaxation of
the MVC problem asks for a value $x_i\geq 0$ for every node $v_i\in V$
such that the sum of all $x_i$ is minimized and such that for every
edge $\set{v_i,v_j}$, $x_i+x_j\geq 1$. A fractional solution can be
turned into an integer solution by rounding up all nodes with a
fractional value at least $1/2$. This increases the approximation
ratio by at most a factor of $2$. Moreover, any maximal matching is a
$2$-approximation for MVC and hence, the randomized parallel algorithm
for maximal matching by Israeli et al. provides a $2$-approximation in
time $O(\log n)$ with high probability \cite{israeli86}. This
indicates that the amount of locality required in order to achieve a
constant approximation for MVC is bounded by $O(\log n)$. In this
section, we present a simple distributed algorithm that places an
upper bound on the achievable trade-off between time complexity and
approximation ratio for the minimum vertex cover problem.

Specifically, the algorithm comes with a parameter $k$, which can
be any integer larger than 0. The algorithm's time
complexity---and hence its locality---is $O(k)$ and its
approximation ratio depends inversely on $k$. The larger $k$, the
better the achieved global approximation.

\IncMargin{0.8em}

\begin{algorithm}[ht]
\SetArgSty{}
    
$x_i \leftarrow 0$;
\lForAll{$e_j\in E_i$}{$y_j \leftarrow 0$}\;

\For{$\ell = k-1,k-2, \ldots,0$}{
    $\tilde{\delta}_i \leftarrow |\{\text{uncovered edges }e\in E_i\}| =
    |\tilde{E}_i|$\;
    
    $\tilde{\delta}_i^{(1)} \leftarrow \max_{i'\in
        \Gamma(v_i)}\tilde{\delta}_{i'}$\;
    
    \If{$\tilde{\delta}_i\ge(\tilde{\delta}_i^{(1)})^{\ell/(\ell+1)}$}
    {
        \lForAll{$e_j\in E_i$}{$y_j\leftarrow y_j + 1/\tilde{\delta}_i$}\;
        $x_i \leftarrow 1$
    }
    
    $Y_i \leftarrow \sum_{e_j\in E_i}y_j$\;
    
    \If{$x_i=0$ \textbf{and} $Y_i\geq 1$}
    {
        \lForAll{$e_j\in E_i$}{$y_j\leftarrow y_j(1+ 1/Y_i)$}\;
        $x_i \leftarrow 1$\;
    }
}

$Y_i \leftarrow \sum_{e_j\in E_i}y_j$\;
\lForAll{$e_j=(v_i,v_{i'})\in E_i$}
{$y_j\leftarrow y_j/\max\{Y_i,Y_{i'}\}$}

\caption{Vertex Cover and Fractional Matching: Code for
  node $v_i\in V$}
\label{alg:mvc_fmm}
\end{algorithm}

\Cref{alg:mvc_fmm} simultaneously approximates both MVC and
its dual problem, the fractional maximum matching (FMM) problem. Let
$E_i$ denote the set of incident edges of node $v_i$. In the FMM
problem, each edge $e_j\in E$ is assigned a value $y_j$ such that the
sum of all $y_j$ is maximized and such that for every node $v_i\in V$,
$\sum_{e_j\in E_i}y_j\leq 1$. The idea of \Cref{alg:mvc_fmm}
is to compute a feasible solution for minimum vertex cover (MVC) and
while doing so, distribute dual values $y_j$ among the incident edges
of each node. Each node $v_i$ that joins the vertex cover $S$ sets its
$x_i$ to $1$ and subsequently, the sum of the dual values $y_j$ of
incident edges $e_j\in E_i$ is increased by $1$ as well. Hence, at the
end of each iteration of the main loop, the invariant $\sum_{v_i\in
  V}{x_i} = \sum_{e_j\in E}{y_j}$ holds. We will show that for all
nodes $v_i$, $\sum_{e_j\in E_i}{y_j} \leq \alpha$ for $\alpha =
3+\Delta^{1/k}$ and that consequently, dividing all $y_j$ by $\alpha$
yields a feasible solution for FMM. By LP duality, $\alpha$ is
an upper bound on the approximation ratio for FMM and MVC. We call an
edge \emph{covered} if at least one of its endpoints has joined the
vertex cover. The set of uncovered edges incident to a node $v_i$ is
denoted by $\tilde{E}_i$, and we define node $v_i$'s \emph{dynamic
  degree} to be $\tilde{\delta}_i:=|\tilde{E}_i|$. The maximum dynamic
degree $\tilde{\delta}_{i'}$ among all neighbors $v_{i'}$ of $v_i$ is
denoted by $\tilde{\delta}_i^{(1)}$.

In the algorithm, a node joins the vertex cover if it has a large
dynamic degree---i.e., many uncovered incident edges---relative to
its neighbors. In this sense, it is a faithful distributed
implementation of the natural sequential greedy algorithm. Because
the running time is limited to $k$ communication rounds, however,
the greedy selection step must inherently be parallelized, even at
the cost of sub-optimal decisions.

The following lemma bounds the resulting decrease of the maximal
dynamic degree in the network.
\begin{lemma}\label{lemma:mvc-anfang}
  At the beginning of each iteration, it holds that
  $\tilde{\delta}_i \leq \Delta ^{(\ell+1)/k}$ for every $v_i\in V$.
\end{lemma}
\begin{proof}
The proof is by induction over the main loop's iterations. For
$\ell=k-1$, the lemma follows from the definition of $\Delta$. For
subsequent iterations, we show that all nodes having
$\tilde{\delta}_i \geq \Delta^{\ell/k}$ set $x_i:=1$ in Line 7. In
the algorithm, all nodes with $\tilde{\delta}_i \geq
(\tilde{\delta}_i^{(1)})^{\ell/(\ell+1)}$ set $x_i:=1$. Hence, we
have to show that for all $v_i$,
$(\tilde{\delta}_i^{(1)})^{\ell/(\ell+1)} \leq \Delta ^{\ell/k}$.
By the induction hypothesis, we know that $\tilde{\delta}_i \leq
\Delta ^{(\ell+1)/k}$ at the beginning of the loop. Since
$\tilde{\delta}_i^{(1)}$ represents the dynamic degree
$\tilde{\delta}_{i'}$ of some node $v_{i'}\in\Gamma(v_i)$, it
holds that $\tilde{\delta}_i^{(1)} \leq \Delta^{(\ell+1)/k}$ for
every such $v_i$ and the claim follows because
$(\tilde{\delta}_i^{(1)})^{\ell/(\ell+1)} \leq \Delta
^{\frac{\ell+1}{k}\cdot \frac{\ell}{\ell+1}}$.
\hspace*{\fill}\end{proof}

The next lemma bounds the sum of dual $y$ values in $E_i$ for an
arbitrary node $v_i\in V$. For that purpose, we define
$Y_i:=\sum_{e_j\in E_i}{y_j}$.
\begin{lemma}\label{lemma:mvc-bound}
  At the end of the algorithm, for all nodes $v_i \in V$,
  \[ Y_i \;=\; \sum_{e_j\in E_i}{y_j} \;\leq\; 3+\Delta^{1/k}.\]
\end{lemma}
\begin{proof}
  Let $\Phi_{h}$ denote the iteration in which $\ell=h$. We distinguish
  three cases, depending on whether (or in which line) a node $v_i$ joins the vertex cover.
  First, consider a node $v_i$ which does not join the vertex cover. Until
  $\Phi_{0}$, it holds that $Y_i<1$ since otherwise, $v_i$ would have set
  $x_i:=1$ in Line~12 of a previous iteration. In $\Phi_{0}$, it
  must hold that $\tilde{\delta}_i=0$ because
  all nodes with $\tilde{\delta}_i\ge1$ set $x_i:=1$ in the last iteration.
  That is, all adjacent nodes $v_{i'}$ of $v_i$
  have set $x_{i'}:=1$ before the last iteration and $Y_i$ does not change
  anymore.  Hence, $Y_i<1$ for nodes which do not belong to the vertex
  cover constructed by the algorithm.

  Next, consider a node $v_i$ that joins the vertex cover in Line 7 of
  an arbitrary iteration $\Phi_{\ell}$. With the same argument as above, we know
  that $Y_i<1$ at the beginning of $\Phi_{\ell}$. When $v_i$ sets
  $x_i:=1$, $Y_i$ increases by one. In the same iteration, however, neighboring
  nodes $v_{i'}\in\Gamma(v_i)$ may also join the vertex cover and thereby
  further increase $Y_i$. By the condition in Line 5, those nodes
  have a dynamic degree at least
  $\tilde{\delta}_{i'}\ge(\tilde{\delta}_{i'}^{(1)})^{\ell/(\ell+1)}
  \ge\tilde{\delta}_i^{\ell/(\ell+1)}$. Further, it holds by 
  \Cref{lemma:mvc-anfang} that $\tilde{\delta}_i\le\Delta^{(\ell+1)/k}$ and
  therefore
  \[
  \tilde{\delta}_i\cdot\frac{1}{\tilde{\delta}_{i'}}\ \le\
  \frac{\tilde{\delta}_i}{\tilde{\delta}_i^{\ell/(\ell+1)}} \;=\;
  \tilde{\delta}_i^{1/(\ell+1)}\ \le\ \Delta^{1/k}.
  \]
  Thus, edges that are simultaneously covered by neighboring nodes may
  entail an additional increase of $Y_i$ by $\Delta^{1/k}$. Together with
  $v_i$'s own cost of $1$ when joining the vertex cover, the total increase of $Y_i$ in Line 6 of
  $\Phi_\ell$ is then at most $1+\Delta^{1/k}$. In Line 6, dual values
  are distributed among uncovered edges only. Therefore, the only way
  $Y_i$ can increase in subsequent iterations is when neighboring
  nodes $x_{i'}$
  set $x_{i'}:=1$ in Line 12. The sum of the $y_j$ of all those edges
  covered only by $v_i$ (note that only these edges are eligible to be
  increased in this way) is at most~1. In Line 11, these $y_j$ can be
  at most doubled. Putting everything together, we have $Y_i\leq 3+\Delta
  ^{1/k}$ for nodes joining the vertex cover in Line 7.

  Finally, we study nodes $v_i$ that join the vertex cover
  in Line 12 of some iteration $\Phi_{\ell}$. Again, it holds that $Y_i<1$ at the
  outset of $\Phi_{\ell}$. Further, using an analogous argument as
  above, $Y_i$ is increased by at most $\Delta^{1/k}$ due to
  neighboring nodes joining the vertex cover in Line 7 of
  $\Phi_{\ell}$. Through the joining of $v_i$, $Y_i$ further
  increases by no more than 1. Because the $y_j$ are increased proportionally, no
  further increase of $Y_i$ is possible. Thus, in this case we have
  $Y_i\leq 2+\Delta ^{1/k}$.
\hspace*{\fill}\end{proof}

Based on the bound obtained in \Cref{lemma:mvc-bound}, the
main theorem follows from LP duality.
\begin{theorem}\label{thm:mvc-approxratio}
  In $k$ rounds of communication, \Cref{alg:mvc_fmm} achieves
  an approximation ratio of $O(\Delta ^{1/k})$. The algorithm is
  deterministic and requires $O(\log \Delta)$ and $O(\log
  \Delta/\log\log \Delta)$ rounds for a constant and polylogarithmic
  approximation, respectively.
\end{theorem}
\begin{proof}
  We first prove that the algorithm computes feasible solutions for
  MVC and fractional maximum matching. For MVC, this is clear because
  in the last iteration, all nodes having $\tilde{\delta}_i\ge1$ set
  $x_i:=1$. The dual $y$-values form a feasible solution because in
  Line 16, the $y_j$ of each edge $e_j$ is divided by the larger of
  the $Y_i$ of the two incident nodes corresponding to $e_j$, and
  hence, all constraints of the fractional matching problem are
  guaranteed to be satisfied.  The algorithm's running time is $O(k)$,
  because every iteration can be implemented with a constant number of
  communication rounds. As for the approximation ratio, it follows
  from \Cref{lemma:mvc-bound} that each $y_j$ is divided by at
  most $\alpha=3+\Delta^{1/k}$ and therefore, the objective functions
  of the primal and the dual problem differ by at most a factor
  $\alpha$. By LP duality, $\alpha$ is a bound on the approximation
  ratio for both problems. Finally, setting $k_1=\beta\log \Delta$ and
  $k_2=\beta\log \Delta/\log\log\Delta$ for an appropriate constant
  $\beta$ leads to a constant and polylogarithmic approximation ratio,
  respectively.
\hspace*{\fill}\end{proof}

Hence, the time-approximation trade-off of Algorithm~\ref{alg:mvc_fmm}
asymptotically nearly matches the lower bound of Theorem
\ref{thm:mvc-lowerbound}. The reason why Algorithm~\ref{alg:mvc_fmm}
does not achieve a constant or polylogarithmic approximation ratio in
a constant number of communication rounds is that it needs to
``discretize'' the greedy step in order to achieve the necessary
parallelism. Whereas the sequential greedy algorithm would select a
single node with maximum dynamic degree in each step, a $k$-local
distributed algorithm must inherently take many such decisions in
parallel. This discrepancy between Algorithm~\ref{alg:mvc_fmm} and the
simple sequential greedy algorithm can be seen even in simple
networks. Consider for instance the network induced by the complete
bipartite graph $K_{m,\sqrt{m}}$. When running the algorithm with
parameter $k=2$, it holds for every node $v_i$ that
$\tilde{\delta}_i\ge(\tilde{\delta}_i^{(1)})^{\ell/(\ell+1)}$ in the
first iteration $(\ell=1)$ of the loop. Hence, every node will join
the vertex cover, resulting in a cover of cardinality
$m+\sqrt{m}$. The optimal solution being $\sqrt{m}$, the resulting
approximation factor is $\sqrt{m}+1=\Delta^{1/2}+1$.

\paragraph*{Remark:} Note that while the lower bound of Theorem
\ref{thm:mvc-lowerbound} holds for the LOCAL model,
Algorithm~\ref{alg:mvc_fmm} does not require the full power of this
model. In particular, to implement Algorithm~\ref{alg:mvc_fmm}, it
suffices to exchange messages containing only $O(\log n)$ bits.

\subsection{Distributed Algorithm for\\ Covering and Packing Linear
  Programs}
\label{sec:lp-upper}

We will now describe a generic distributed algorithm to solve covering
and packing LPs in the network setting described in Section
\ref{sec:problems}. The algorithm is based on a randomized technique
to cover a graph with clusters of small diameter described in
\cite{linial93}. The property of covering and packing LPs allows to
solve local sub-LPs for all clusters and to combine the local
solutions into an approximate global one.

Assume that we are given a primal-dual pair of covering and packing
LPs of the canonical form (P) and (D) and the corresponding network
graph $G_{\lp}=(V_p\dot{\cup} V_d,E)$ as defined in Section
\ref{sec:problems}. We first describe how the local sub-LPs look
like. Let $Y=\set{y_1,\dots,y_{n_d}}$ be the set of variables of
(D). Each local primal-dual sub-LP pair is defined by a subset
$S\subseteq V_d$ of the dual nodes $V_d$ and thus by a subset
$Y_S\subseteq Y$ of the dual variables. There is a one-to-one
correspondence between the inequalities of (P) and the variables of
(D). Let $P_S$ be the LP that is obtained from (P) by restricting to
the inequalities corresponding to the dual variables in $Y_S$.  The
primal variables involved in $P_S$ are exactly the ones held by primal
nodes $V_p$ that are neighbors of some node in $S$. The local LP $D_S$
is the dual LP of $P_S$.  The variables of $D_S$ are given by the set
of inequalities of $P_S$ and therefore by $Y_S$. We first prove
crucial basic properties of such a pair of local sub-LPs.

\begin{lemma}\label{lemma:subLPs}
  Assume that we are given a pair of LPs $P_S$ and $D_S$ that are
  constructed from (P) and (D) as described above. If (P) and (D) are
  both feasible, $P_S$ and $D_S$ are both feasible. Further, any
  solution to $D_S$ (with dual variables in $Y\setminus Y_S$ set to
  $0$) is a feasible solution of (D).
\end{lemma}
\begin{proof}
  Clearly $P_S$ is feasible as every feasible solution for (P)
  directly gives a feasible solution for $P_S$ (by just ignoring all
  variables that do not occur in $P_S$). Because $P_S$ is a
  minimization problem and since (P) and (D) are covering and packing
  LPs, all coefficients in the objective function (the vector
  $\vektor{c}$) are non-negative, $P_S$ is also bounded (its objective
  function is always at least $0$). Hence, also the dual LP $D_S$ must
  be feasible.

  Assume that we are given a feasible solution for $D_S$ which is
  extended to a solution for (D) by setting variables in $Y\setminus
  Y_S$ to $0$. Inequalities in (D) that correspond to columns of
  variables occurring in $P_S$ are satisfied by the feasibility of the
  solution for $D_S$. In all other inequalities of (D), all variables
  are set to $0$ and thus, feasibility follows from the fact that
  $\vektor{c}\geq 0$.
\hspace*{\fill}\end{proof}

Note that by construction, a feasible solution for $P_S$ gives a
solution for (P) which satisfies all the inequalities corresponding to
variables $Y_S$ and for which the left-hand sides of all other
inequalities are at least $0$ because all coefficients and variables
are non-negative. We next show how to obtain local sub-LPs $P_S$ and
$D_S$ that can be solved efficiently.

In \cite{linial93}, Linial and Saks presented a randomized distributed
algorithm for a weak-diameter network decomposition. We use their
algorithm to decompose the linear program into sub-programs which can
be solved locally in the LOCAL model. Assume that we are
given a network graph $\mathcal{G}=(\mathcal{V},\mathcal{E})$ with
$n=|\mathcal{V}|$ nodes. The basic building block of the algorithm in
\cite{linial93} is a randomized algorithm $\LS(p,R)$ which
computes a subset $\mathcal{S}\subseteq\mathcal{V}$ such that each
node $u\in \mathcal{S}$ has a leader $\ell(u)\in \mathcal{V}$ and the
following properties hold for arbitrary parameters $p\in[0,1]$ and
$R\ge 1$:

\begin{enumerate}
\item $\forall u\in\mathcal{S}:\ d_{\mathcal{G}}(u,\ell(u))\le R$,
  where $d_{\mathcal{G}}(u,v)$ is the shortest path distance between
  two nodes $u,v\in\mathcal{V}$.
\item $\forall u,v\in\mathcal{S}:\ \ell(u)\neq\ell(v)\
  \Longrightarrow\ (u,v)\not\in\mathcal{E}$.
\item $\mathcal{S}$ can be computed in $\Oh(R)$ rounds.
\item $\forall u\in\mathcal{V}:\ \Pr[u\in
    \mathcal{S}]\ge p(1-p^R)^{n-1}$.
\end{enumerate}

Hence, Algorithm $\LS(p,R)$ computes a set of clusters of nodes such
that nodes belonging to different clusters are at distance at least
$2$ and such that every node $u$ that belongs to some cluster is at
distance at most $R$ from its cluster center $\ell(u)$. Note that
Algorithm $\LS(p,R)$ does bound the distance between nodes of the same
cluster in the graph induced by the nodes of the cluster. It merely
bounds their distance in $\mathcal{G}$. The maximal
$\mathcal{G}$-distance between any two nodes of a cluster is called
the \emph{weak diameter} of the cluster.

Based on the graph $G_{\lp}$, we define the graph
$\mathcal{G}=(\mathcal{V},\mathcal{E})$ on which we invoke Algorithm
$\LS(p,R)$:
\begin{equation}\label{eq:LSgraph}
  \mathcal{V}:=V_d,\quad\mathcal{E}:=\set{\set{u,v}\in{V_s\choose
      2}\bigg|d(u,v)\leq 4},
\end{equation}
where $d(u,v)$ denotes the distance between $u$ and $v$ in $G_{\lp}$.
Hence, the nodes of $\mathcal{G}$ are all nodes corresponding to dual
variables in $G_{\lp}$. As discussed, there is a one-to-one
correspondence between nodes in $V_d$ and inequalities in the linear
program (P). Two nodes $u,v\in V_d$ are connected by an edge in
$\mathcal{E}$ iff the corresponding inequalities contain variables
that occur together in some inequality.  We apply Algorithm $\LS(p,R)$
several times on graph $\mathcal{G}$ to obtain different locally
solvable sub-LPs that can then be combined into an approximate
solution for (P) and (D). The details are given by Algorithm
\ref{alg:lpalg}.

\begin{algorithm}[h] 
\setlength{\parskip}{0.4ex}

Run $\ell$ independent instances of $\LS(p,R)$ on $\mathcal{G}$
\textbf{in parallel}:\\
\hspace*{5mm}  yields node sets
$\mathcal{S}_1,\dots,\mathcal{S}_\ell\subseteq \mathcal{V}=V_d$\;

Solve local LPs $P_{\mathcal{S}_1}, D_{\mathcal{S}_1}, \dots,
P_{\mathcal{S}_\ell}, D_{\mathcal{S}_\ell}$\;

Interpret as solutions for (P) and (D):
\hspace*{5mm} $x_{1,1},\dots,x_{1,n_p},y_{1,1},\dots,y_{1,n_d},\dots
x_{\ell,1},\dots,x_{\ell,n_p},y_{\ell,1},\dots,y_{\ell,n_d}$\;

\lForAll{$i\in \set{1,\dots,n_p}$}{$x_i \leftarrow \sum_{t=1}^\ell
  x_{t,i}$}\;
\lForAll{$i\in \set{1,\dots,n_d}$}{$y_i \leftarrow \sum_{t=1}^\ell
  y_{t,i}$}\;

\lForAll{$i\in \set{1,\dots,n_p}$}
{$x_i \leftarrow x_i/\min_{v_j^d\in\Gamma_{v_i^{p}}}(A\vektor{x})_j/b_j$}\;
\lForAll{$i\in \set{1,\dots,n_d}$}{
  $y_i \leftarrow y_i/\ell$}\;

\Return $\vektor{x}$ and $\vektor{y}$

\caption{Algorithm for Covering and Packing linear programs with parameters: $\ell$,
  $p$, and $R$}
\label{alg:lpalg}
\end{algorithm}

We first analyze the time complexity of Algorithm \ref{alg:lpalg}.

\begin{lemma}\label{lemma:lpalgtime}
  Algorithm \ref{alg:lpalg} can be executed in $\Oh(R)$ rounds. It
  computes feasible solutions for (P) and (D).
\end{lemma}
\begin{proof}
  Algorithm \ref{alg:lpalg} consists of the following main
  steps. First, $\ell$ independent instances of Algorithm $\LS(p,R)$ are
  executed. Then, for each collection of clusters resulting from
  these executions, a local LP is solved and the local LPs are
  combined to solutions of (P) and (D). Finally, each resulting
  dual variable is divided by $\ell$ and each primal variable is divided
  by an amount that keeps the primal solution feasible.

  As the $\ell$ instances of Algorithm $\LS(p,R)$ are independent, they
  can be executed in parallel and thus the time complexity for the
  first step is $\Oh(R)$. Note that since neighbors in $\mathcal{G}$
  are at distance at most $4$, each round on $\mathcal{G}$ can be
  simulated in $4$ rounds on $G_{\lp}$.

  For the second step, consider the set of nodes
  $\mathcal{S}_i\subseteq V_d$ computed by the $i^{\mathit{th}}$
  instance of $\LS(p,R)$. Two nodes that are in different connected
  components of the sub-graph $\mathcal{G}[\mathcal{S}_i]$ of
  $\mathcal{G}$ induced by $\mathcal{S}_i$ are at distance at least
  $5$. Hence, the corresponding dual variables and also the primal
  variables corresponding to their $G_{\lp}$-neighbors in $V_p$ cannot
  occur together in an inequality of (D) and (P), respectively. Hence,
  the local sub-LP induced by $\mathcal{S}_i$ can be solved by
  individually solving the sub-LPs induced by every connected
  component of $\mathcal{G}[\mathcal{S}_i]$. As every connected
  component of $\mathcal{G}[\mathcal{S}_i]$ has a leader node that is
  at distance at most $R$ from all nodes of the connected component,
  all information from the sub-LP corresponding to
  $\mathcal{G}[\mathcal{S}_i]$ can be sent to this leader and the
  sub-LP can be solved there locally. Hence, the second step of the
  algorithm can also be executed in $\Oh(R)$ rounds.

  For the third step, note that the values by which the primal
  variables $x_i$ are divided can be computed locally (by only
  exchanging information with direct neighbors in $G_{\lp}$). Finally,
  the computed dual solution is feasible because it is the average of
  the dual solutions of all sub-LP and because each dual sub-LP is
  feasible for (D) by Lemma \ref{lemma:subLPs}. Line 7 of Algorithm
  \ref{alg:lpalg} guarantees that the computed primal solution is a
  feasible solution for (P).
\hspace*{\fill}\end{proof}

\begin{theorem}\label{thm:lpalg}
  Let $\eps\in(0,1)$, $\alpha>1$, and $\beta>0$ be parameters. We
  choose $p=n_d^{-\alpha/R}$ and define $q:=p\cdot(1-n_d\cdot
  p^R)$. If we choose $\ell\geq\frac{2(1+\beta)}{\eps^2q}\ln n_d$,
  Algorithm \ref{alg:lpalg} computes $\frac{1}{q(1-\eps)}$
  approximations for (P) and (D) in $\Oh(R)$ time (in the LOCAL model)
  with probability at least $1-1/n_d^\beta$.
\end{theorem}
\begin{proof}
  The time complexity follows directly from Lemma
  \ref{lemma:lpalgtime}. Because by Lemma \ref{lemma:lpalgtime}, the
  computed solutions for (P) and (D) are both feasible, the
  approximation ratio of the algorithm is bounded by the ratio of the
  objective functions of the solutions for (P) and (D). Both solutions
  are computed as the sum of the solutions of all local sub-LPs in
  Lines 5 and 6 of the algorithm that are then divided by
  appropriate factors in Lines 7 and 8. By LP duality (of the
  sub-LPs), we have $\vektor{c}^T\vektor{x}=\vektor{b}^T\vektor{y}$
  after Line 6. Hence, the approximation ratio is upper bounded be the
  ratio between the factor $\ell$ by which the dual variables are
  divided and the minimal value by which the primal variables are
  divided. The approximation is therefore upper bounded by
  \begin{equation}\label{eq:lpalgapproximation}
    \frac {\ell}{\min_{v_i^p\in V_p}\min_{v_j^d\in\Gamma_{v_i^p}}(A\vektor{x})_j/b_j}
    =
    \frac {\ell}{\min_{v_j^d\in V_d}(A\vektor{x})_j/b_j}
 \end{equation}
 for $\vektor{x}$ after Line 6.  To obtain an upper bound on the value
 of the above equation, assume that for every $v_j^d\in V_d$, the
 number of local sub-LPs $P_{\mathcal{S}_t}$ for which
 $(A\vektor{x_t})_j\geq b_j$ is at least $\ell'\leq\ell$. Hence,
 $\ell'$ is a lower bound on the number of times each inequality of
 (P) is satisfied, combined over all sub-LPs. Because $b_j\geq 0$ for
 all $j$ and because all coefficients of $A$ and the variables $x$ are
 non-negative, we then have $(A\vektor{x})_j/b_j\geq \ell'$ for all
 $j$. By \Cref{eq:lpalgapproximation}, it then follows
 that the computed solutions for (P) and (D) are at most by a factor
 $\ell/\ell'$ worse than the optimal solutions.

 We get a bound on the minimum number of times each inequality of (P)
 is satisfied by a local sub-LP by using the properties of Algorithm
 $\LS(p,R)$ and a Chernoff bound. From \cite{linial93}, we have that
 for each $t\in\set{1,\dots,\ell}$ and $v_i^d\in V_d$, the probability
 that $v_i^d\in \mathcal{S}_t$ is at least

\[
p(1-p^R)^{n_d-1} = \stackrel{p^R} =
\frac{1}{n_d^{\alpha/R}}\cdot\left(1-\frac{1}{n_d^\alpha}\right)^{n_d-1}
\stackrel{(\alpha>1)}{\geq}
\frac{1}{n_d^{\alpha/R}}\cdot\left(1-\frac{1}{n_d^{\alpha-1}}\right) = q.
\]

Therefore, for every $v_j^d\in V_d$, the probability $P_j$ that
$j^{\mathit{th}}$ inequality of (P) is satisfied less than
$(1-\eps)q\ell$ times is at most
\begin{equation}\label{eq:lpalgcoverage}
  P_j < e^{-\frac{\eps^2}{2}q\ell} \leq e^{-(1+\beta)\ln n_d} =
  \frac{1}{n_d}\cdot\frac{1}{n^{\beta}}.
\end{equation}
The theorem now follows by a union bound over all $n_d$ inequalities
of (P).
\hspace*{\fill}\end{proof}

\begin{corollary}
  In $k$ rounds in the LOCAL model, Algorithm \ref{alg:lpalg} with
  high probability computes an $n^{c/k}$-approximation for covering
  and packing LPs for some constant $c>0$. An $(1+\eps)$-approximation
  can be computed in time $\Oh(\log(n)/\eps)$.
\end{corollary}

\subsection{Randomized Rounding}
\label{sec:rounding}

We next show how to use the algorithm of the last section to solve the MDS
problem or another combinatorial covering or packing problem. Hence,
we show how to turn a fractional covering or packing solution into an
integer one by a distributed rounding algorithm. In particular, we
give an algorithm for integer covering and packing problems of the
forms

\hspace*{\fill}\parbox{45mm}{
  \begin{align*}
    \min && \vektor{c}^\mathrm{T}\vektor{x}' & \\
    \text{s.\ t.} && A\cdot\vektor{x}' & \ge \vektor{b} \\
    && x_i' & \in \mathbb{N}.
  \end{align*}}\hfill
\parbox{1cm}{\hfill(P$_I$)}\ \
\hspace*{\fill}\parbox{45mm}{
  \begin{align*}
    \min && \vektor{b}^\mathrm{T}\vektor{y}' & \\
    \text{s.\ t.} && A^{\mathrm{T}}\cdot\vektor{y}' & \le \vektor{c} \\
    && y_i' & \in \mathbb{N}.
  \end{align*}}\hfill
\parbox{1cm}{\hfill(D$_I$)}

\noindent with matrix elements $a_{ij}\in\{0,1\}$. LPs (P) and (D) are
the fractional relaxations of (P$_I$) and (D$_I$). Not that we denote
the solution vectors for the integer program by $\vektor{x}'$ and
$\vektor{y}'$ whereas the solution vectors for the corresponding LPs
are called $\vektor{x}$ and $\vektor{y}$.

We start with covering problems (problems of the form of (P)).
Because the $a_{ij}$ and the $x_i$ are restricted to integer values,
w.l.o.g.\ we can round up all $b_j$ to the next integer value. After
solving/approximating the LP, each primal node $v_i^{p}$ executes
Algorithm \ref{alg:roundcover}. The value of the parameter $\lambda$
will be determined later.

\begin{algorithm}[H]
  \SetArgSty{}

  \eIf{$x_i\geq 1(\lambda\ln\Delta_p)$}
  {
    $x_i' \leftarrow \lceil x_i\rceil$
  }
  {
    $p_i \leftarrow x_i\cdot\lambda\ln\Delta_p$\;
    $x_i' \leftarrow 1$ with probability $p_i$ and $x_i'\leftarrow 0$ otherwise
  }

  \caption{Distributed Randomized Rouding: Covering Problems}
  \label{alg:roundcover}
\end{algorithm}

The expected value of the objective function is
$\mathrm{E}[\vektor{c}^\mathrm{T}\vektor{x}']\le\lambda\ln\Delta_p\cdot
\vektor{c}^\mathrm{T}\vektor{x}$. Yet regardless of how we choose
$\lambda$, there remains a non-zero probability that the obtained
integer solution is not feasible. To overcome this, we have to
increase some of the $x_i'$. Assume that the $j^\mathrm{th}$
constraint is not satisfied. Let $\vektor{a}_j$ be the row vector
representing the $j^\mathrm{th}$ row of the matrix $A$ and let
$b_j':=b_j-\vektor{a}_i\vektor{x}'$ be the missing weight to make the
$j^\mathrm{th}$ row feasible. Further, let $\ijmin$ be the index of
the minimum $c_i$ for which $a_{ji}=1$.  We set
$x_{\ijmin}':=x_{\ijmin}'+b_j'$. Applied to all non-satisfied primal
constraints, this gives a feasible solution for the considered integer
covering problem.

\begin{theorem}\label{thm:intcovering}
  Consider an integer covering problem (P$_I$) with $a_{ij}=\{0,1\}$
  and $b_j\in\mathbb{N}$.  Furthermore, let $\vektor{x}$ be an
  $\alpha$-approximate solution for the LP relaxation (P) of
  (P$_I$). The above described algorithm computes an
  $\Oh(\alpha\log\Delta_p)$-approximation $\vektor{x}'$ for (P$_I$) in
  a constant number of rounds.
\end{theorem}
\begin{proof}
  As stated above, the expected approximation ratio of the first part
  of the algorithm is $\lambda\ln\Delta_p$. In order to bound the
  additional weight of the second part, where $x_{\ijmin}'$ is
  increased by $b_j'$, we define dual variables
  $\tilde{y}_j:=b_j'c_{\ijmin}/b_j$. For each unsatisfied primal
  constraint, the increase $c_{\ijmin}b_j'$ of the primal objective
  function is equal to the increase $b_j\tilde{y}_j$ of the dual
  objective function. If the $j^\mathrm{th}$ constraint is not
  satisfied, we have $b_j'\ge1$. Therefore,
  $\mathrm{E}[\tilde{y}_j]\le q_j c_{\ijmin}$, where $q_j$ is the
  probability that the $j^\mathrm{th}$ primal inequality is not
  fulfilled.

  In order to get an upper bound on the probability $q_j$, we have to
  look at the sum of the $x_i'$ before the randomized rounding step in
  Line 5 of the algorithm. Let $\beta_j:=b_i-\vektor{a}_i\vektor{x'}$
  be the missing weight in row $j$ before Line 5. Because the
  $x$-values correspond to a feasible solution for the LP, the sum of
  the $p_i$ involved in row $j$ is at least
  $\beta_j\lambda\ln\Delta_p$. For the following analysis, we assume
  that $\ln\Delta_p\ge1$. If $\ln\Delta_p<1$, applying only the last
  step of the described algorithm gives a simple distributed
  2-approximation for the considered integer program. Using a
  Chernoff bound, we can bound $q_j$ as
  \begin{equation*}
    q_j <
    e^{-\frac{1}{2}\beta_j\lambda\ln\Delta_p
      (1-\frac{1}{\lambda\ln\Delta_p})^2}
    \le \left(\frac{1}{\Delta_p}\right)^
    {\frac{1}{2}\lambda(1-\frac{1}{\lambda})^2} \!\!\!\!\!\!\le
    \frac{1}{\Delta_p}.
  \end{equation*}
  In the second inequality, we use that $\beta_j\ge1$.  For the last
  inequality, we have to choose $\lambda$ such that
  $\lambda(1-1/\lambda)^2/2\ge1$ (i.e., $\lambda\ge2+\sqrt{3}$). Thus,
  the expected value of $\tilde{y}_j$ is $\mathrm{E}[\tilde{y}_j]\le
  c_{\ijmin}/\Delta_p$.  Hence, by definition of $c_{\ijmin}$, in
  expectation the $\tilde{y}$-values form a feasible solution for
  (D). Therefore, the expected increase of the objective function
  $\vektor{c}^\mathrm{T}\vektor{x}'$ in the last step after the
  randomized rounding is upper-bounded by the objective function of an
  optimal solution for (P).
\hspace*{\fill}\end{proof}

Combining Algorithms \ref{alg:lpalg} and \ref{alg:roundcover}, we
obtain an $\Oh(\log\Delta)$-approximation for MDS in $\Oh(\log n)$
rounds.

We now turn our attention to integer packing problems. We have an
integer program of the form of (D$_I$) where all $a_{ij}\in\{0,1\}$ and
where $\vektor{y}'\in\mathbb{N}^n$. We can w.l.o.g.\ assume that the
$c_j$ are integers because rounding down each $c_j$ to the next
integer has no influence on the feasible region. Each dual node
$v_i^{d}$ applies Algorithm \ref{alg:roundpacking}.

\begin{algorithm}[H]
  \SetArgSty{}

  \eIf{$y_i\geq 1$}
  {
    $y_i' \leftarrow \lfloor y_i\rfloor$
  }
  {
    $p_i \leftarrow 1/(2e\Delta_d)$\;
    $y_i' \leftarrow 1$ with probability $p_i$ and $y_i'\leftarrow 0$ otherwise
  }

  \If{$y_i'\in$ `non-satisfied constraint'}
  {
    $y_i' \leftarrow \lfloor y_i\rfloor$
  }

  \caption{Distributed Randomized Rouding: Packing Problems}
  \label{alg:roundpacking}
\end{algorithm}

Clearly, the yields a feasible solution for the problem. The
approximation ratio of the algorithm is given by the next theorem.
\begin{theorem}\label{thm:intpacking}
  Let (D$_I$) be an integer covering problem with $a_{ij}=\{0,1\}$ and
  $c_j\in\mathbb{N}$.  Furthermore, let $\vektor{y}$ be an
  $\alpha$-approximate solution for the LP relaxation of
  (D$_I$). Algorithm \ref{alg:roundpacking} computes an
  $\Oh(\alpha\Delta_d)$-approximation $\vektor{y}'$ for (D$_I$) in a
  constant number of rounds.
\end{theorem}
\begin{proof}
  After Line 6, the expected value of the objective function is
  $\vektor{b}^\mathrm{T}\vektor{y}' \ge
  \vektor{b}^\mathrm{T}\vektor{y}/(2e\Delta_d)$. We will now show that
  a non-zero $y_i'$ stays non-zero with constant probability in Line
  8.  Let $q_j$ be the probability that the $j^\mathrm{th}$ constraint
  of the integer program is not satisfied given that $y_i'$ has been
  set to 1 in Line 5. For convenience, we define $Y_j':=\sum_i
  a_{ij}y_i'$. If $c_j\ge2$, we apply a Chernoff bound to obtain
  \begin{eqnarray*}
    q_j& =& \prob[Y_j'>c_j\,\big|\,y_i'=1]\ \le\
    \prob[Y_j'>c_j-1] \\
    &<& \left(\frac{e^{e\Delta_d-1}}
      {(e\Delta_c)^{e\Delta_d}}\right)^{c_j/(2e\Delta_d)}
    < \frac{1}{\Delta_d}.
  \end{eqnarray*}
  If $c_j=1$, we get
  \begin{eqnarray*}
    q_j &\le& 1-\prob[Y_j'=0]\ =\
    1-\!\!\!\!\prod_{v_i^d\in\Gamma(v_j^p)}\!\!\!\!\left(1-p_i\right)\\
    &\le& 1-\left(1-\frac{1}{2e\Delta_d}\right)\ =\ \frac{1}{2e\Delta_d}.
  \end{eqnarray*}
  The probability that all dual constraints containing $y_i'$ are
  satisfied is lower-bounded by the product of the probabilities for
  each constraint \cite{srinivasan95}. Therefore, under the natural
  assumption that $\Delta_d\ge2$:
  \[
  \prob[y_i'=1\text{ after Line 8}] \ge
  \left(1-\frac{1}{\Delta_d}\right)^{\Delta_d} \ge \frac{1}{4}.
  \]
  Thus the expected value of the objective function of the integer
  program (D$_I$) is
  \[
  \mathrm{E}[\vektor{b}^\mathrm{T}\vektor{y}']\
  \ge\ 8e\Delta_d\cdot\vektor{b}^\mathrm{T}\vektor{y}.
  \]
\hspace*{\fill}\end{proof}

\paragraph*{Remark} As stated, Algorithms \ref{alg:roundcover} and
\ref{alg:roundpacking} require the nodes to know the maximum primal
and dual degrees $\Delta_p$ and $\Delta_d$, respectively. In both
cases, it would be possible to replace the use of $\Delta_p$ and
$\Delta_d$ by local estimates of these quantities. In order to keep
the algorithms and the analysis as simple as possible, we decided to
state and analyze them in the present form.

\subsection{Connecting a Dominating Set}
\label{sec:cds-upper}

An important applications of dominating sets in networks is to obtain
clusterings in ad hoc or sensor networks. In particular, clustering
helps to improve information dissemination and routing algorithms in
such networks. However, for this purpose, one usually needs clusters
to be connected to each other and thus a connected dominating set as
underlying structure. Lemma \ref{lemma:distcds} in Section
\ref{sec:cdslower} shows that every dominating set $D$ can be extended
to a connected dominating set $D'$ of size $|D'|<3|D|$. In the
following, we described a simple distributed strategy to convert any
dominating set into a connected dominating set that is only slightly
larger. A similar strategy is also used in \cite{dubhashi03}.

Assume that we are given a dominating set $D$ of the network graph
$G$. As in the proof of Lemma \ref{lemma:distcds}, we define a graph
$G_D$ as follows. The node set of $G_D$ is $D$ and there is an edge
between $u,v\in D$ iff their distance in $G$ is at most $3$. We have
seen that $G_D$ is connected and thus, any spanning tree of
$G_D$induces a connected dominating set of size $\Oh(D)$.
Unfortunately, for a local, distributed algorithm, it is not possible
to compute a spanning tree of $G_D$.  Nevertheless, a similar approach
also works for distributed algorithms. Instead of computing a spanning
tree of $G_D$, it is sufficient to compute any sparse spanning
subgraph of $G_D$. If the number of edges of the subgraph of $G_D$ is
linear in the number of nodes $|D|$ of $G_D$, we obtain a connected
dominating set $S'$ which is only by a constant factor larger than
$D$.

We therefore need to solve the following problem. Given a graph
$G=(V,E)$ with $|V|=n$, we want to compute a spanning subgraph $G'$ of
$G$ with a minimal number of edges. For an arbitrary $k\ge1$, the
following Algorithm \ref{alg:distcds} shows how to compute such a
spanning subgraph in $k$ rounds. For the algorithm, we assume that all
edges $e=(u,v)$ of $G$ have a unique weight $w_e$ and that there is a
total order on all edge weights. If there are no natural edge weights,
a weight for $(u,v)$ can for example be constructed by taking the
ordered pair of the IDs of the nodes $u$ and $v$. Two weights can be
compared using lexicographic order.

\begin{algorithm}[H]
  \SetArgSty{}

  $G' \leftarrow G$\;
  \lForAll{$u\in V$}{ $u$ \textbf{collects} complete $k$-neighborhood}\;

  \ForAll{$e\in E$}{
    \If{weight $w_e$ of $e$ is largest in any cycle of length $\leq
      2k$}
    { remove $e$ from $G$}
  }
  \caption{Computing a sparse connected subgraph}
  \label{alg:distcds}
\end{algorithm}

The following lemma shows that Algorithm \ref{alg:distcds} indeed
computes a sparse connected subgraph $G'$ of $G$.

\begin{lemma}\label{lemma:nofedges}
  For every $n$-node connected graph $G=(V,E)$ and every $k$,
  Algorithm \ref{alg:distcds} computes a spanning subgraph
  $G'=(V,E')$ of $G$ for which the number of edges is bounded by
  $|E'|\le n^{1+O(1/k)}$.
\end{lemma}
\begin{proof} 
  We first prove that the produced $G'$ is connected. For the sake of
  contradiction, assume that $G'$ is not connected. Then, there must
  be a cut $(S,T)$ with $S\subseteq V$, $T=V\setminus S$, and
  $S,T\neq\emptyset$ such that $S\times T\cap E'=\emptyset$. However,
  since $G$ is connected, there must be an edge $e\in S\times T\cap E$
  crossing the given cut. Let $e$ be the edge with minimal weight
  among all edges crossing the cut. Edge $e$ can only be removed by
  Algorithm \ref{alg:distcds} if it has the largest weight of all
  edges in some cycle. However, all cycles containing $e$ also contain
  another edge $e'$ crossing the $(S,T)$-cut. By definition of $e$,
  $w_{e'}>w_e$ and therefore, $e$ is not deleted by the algorithm.

  Let us now look at the number of edges of $G'$. Because in every
  cycle of length at most $2k$ at least one edge is removed by
  Algorithm \ref{alg:distcds}, $G'$ has girth $g(G')\ge2k+1$. It is
  well-known that therefore, $G'$ has at most $|V|^{1+\Oh(1/k)}$
  edges (see e.g.~\cite{bollobas-book}).
\hspace*{\fill}\end{proof}

We can therefore formulate a $k$-round MCDS algorithm consisting of
the following three phases. First, a fractional dominating set is
computed using Algorithm \ref{alg:lpalg}. Second, we use the
randomized rounding scheme given by Algorithm \ref{alg:roundcover} to
obtain a dominating set $D$. Finally, Algorithm \ref{alg:distcds} is
applied to $G_D$. For each edge $(u,v)$ of the produced spanning
subgraph of $G_D$, we add the nodes (at most $2$) of a shortest path
connecting $u$ and $v$ in $G$ to $D$. Note that a $k$-round algorithm
on $G_D$ needs at most $3k$ rounds when executed on $G$. The achieved
approximation ratio is given by the following theorem.
\begin{theorem}\label{thm:distcds}
  In $\Oh(k)$ rounds, the above described MCDS algorithm computes a
  connected dominating set of expected size
  \[
  \Oh\!\left(\mathrm{CDS_{OPT}}\cdot n^{\Oh(1/k)}\cdot\log\Delta\right).
  \]
\end{theorem}
\begin{proof}
  Given the dominating set $D$, by Lemma \ref{lemma:nofedges}, the
  number of nodes of the connected dominating set $D'$ can be bounded
  by
  \[
  |D'|\ \le\ 3|D|^{1+\Oh(1/k)}\le\ 3|D|n^{\Oh(1/k)}
  \]
  and therefore
  \begin{equation}\label{eq:sizeCDS}
    \E[|D'|]\ \le\ 3\E[|D|]n^{\Oh(1/k)}.
  \end{equation}
  Using Theorems \ref{thm:lpalg} and \ref{thm:intcovering}, it
  follows that the expected size of the dominating set $D$ is
  \begin{equation*}
  \E[|D|]\ \in\ \Oh\big(\mathrm{DS_{OPT}}n^{\Oh(1/k)}\log\Delta\big).
  \end{equation*}
  Plugging this into \Cref{eq:sizeCDS} completes the
  proof.
\hspace*{\fill}\end{proof}

\subsection{Role of Randomization and Distributed Derandomization}
\label{sec:randomization}

Randomization plays a crucial role in distributed algorithms. For
many problems such as computing a MIS, there are simple and
efficient randomized algorithms. For the same problems, the best
deterministic algorithms are much more complicated and usually
significantly slower. The most important use of randomization in
distributed algorithms is breaking symmetries. We have seen that in
certain cases, LP relaxation can be used to ``avoid'' symmetry
breaking.  The question is whether the use of randomness can also be
avoided in such cases?
In the following, we show that this indeed is the case, i.e., we
show that in the LOCAL model \emph{any distributed randomized
algorithm for solving a linear program can be derandomized.}

Assume that we are given a randomized distributed $k$-round algorithm
$\mathcal{A}$ which computes a solution for an arbitrary linear
program $P$. We assume that $\mathcal{A}$ explicitly solves $P$ such
that w.l.o.g.\ we can assume that each variable $x_i$ of $P$ is
associated with a node $v$ which computes $x_i$. We also assume that
$\mathcal{A}$ always terminates with a feasible solution. The
following theorem shows that $\mathcal{A}$ can be derandomized.

\begin{theorem}\label{thm:derand}
  Algorithm $\mathcal{A}$ can be
  transformed into a deterministic $k$-round algorithm $\mathcal{A}'$
  for solving $P$. The objective value of the solution produced by
  $\mathcal{A'}$ is equal to the expected objective value of the
  solution computed by $\mathcal{A}$.
\end{theorem}
\begin{proof}
  We first show that for the node computing the
  value of variable $x_i$, it is possible to deterministically compute
  the expected value $\E[x_i]$. We have seen that in the LOCAL
  model every deterministic $k$-round algorithm can be formulated as
  follows. First, every node collects all information up to distance
  $k$. Then, each node computes its output based on this information.
  The same technique can also be applied for randomized algorithms.
  First, every node computes all its random bits. Collecting the
  $k$-neighborhood then also includes collecting the random bits of
  all nodes in the $k$-neighborhood. However, instead of computing
  $x_i$ as a function of the collected information (including the
  random bits), we can also compute $\E[x_i]$ without even
  knowing the random bits.

  In algorithm $\mathcal{A'}$, the value of each variable is now set
  to the computed expected value. By linearity of expectation, the
  objective value of $\mathcal{A'}$'s solution is equal to the
  expected objective value of the solution of $\mathcal{A}$. It
  remains to prove that the computed solution is feasible. For the
  sake of contradiction, assume that this is not the case. Then, there
  must be an inequality of $P$ which is not satisfied. By linearity of
  expectation, this implies that this inequality is not satisfied in
  expectation for the randomized algorithm $\mathcal{A}$.  Therefore,
  there is a non-zero probability that $\mathcal{A}$ does not fulfill
  the given inequality, a contradiction to the assumption that
  $\mathcal{A}$ always computes a feasible solution.
\hspace*{\fill}\end{proof}

Theorem \ref{thm:derand} implies that the algorithm of Section
\ref{sec:lp-upper} could be derandomized to deterministically compute an
$(1+\eps)$-approximation for (P) and (D) in $\Oh(\log(n)/\eps)$
rounds.  It also means that in principle every distributed dominating
set algorithm (e.g.\ \cite{jia-podc01,rajagopalan98} could be turned into a
deterministic fractional dominating set algorithm with the same
approximation ratio. Hence, when solving integer linear programs in
the LOCAL model, randomization is only needed to break
symmetries. Note that this is really a property of the LOCAL
model and only true as long as there is no bound on message sizes and
local computations. The technique described in Theorem
\ref{thm:derand} can drastically increase message sizes and local
computations of a randomized distributed algorithm.


%% file: conclusions.tex
\section{Conclusions \& Future Work}\label{sec:conclusion}

\textbf{Lower Bounds: }Distributed systems is an area in computer
science with a strong lower bound culture. This is no coincidence as
lower bounds can be proved using \emph{indistinguishability}
arguments, i.e. that some nodes in the system cannot distinguish two
configurations, and therefore must make ``wrong'' decisions.

Indistinguishability arguments have also been used in locality. In
his seminal paper, Linial proved an $\Omega(\log^* \! n)$ lower
bound for coloring the ring topology
~\cite{linial92}. However, one cannot prove local inapproximability
bounds on the ring or other highly symmetric topologies, as they allow
for straight-forward purely local, constant approximation
solutions. Take for instance the minimum vertex cover problem (MVC):
In any $\delta$-regular graph, the algorithm which includes all nodes
in the vertex cover is already a 2-approximation. Each node will cover
at most $\delta$ edges, the graph has $n\delta/2$ edges, and therefore
at least $n/2$ nodes need to be in a vertex cover.

Further, also several natural asymmetric graph families enjoy
constant-time algorithms. For example, in a tree, choosing all inner
nodes yields a 2-approximation MVC. More generally, a similar
algorithm also yields a constant MVC approximation for arbitrary
graphs with bounded arboricity (i.e., graphs where all subgraphs are
sparse, includes minor-closed families such as planar graphs). On the
other extreme, also very dense graph classes have very efficient MVC
algorithms. Graphs from such families often have small diameter and in
addition, as each node can cover at most $n-1$ edges, in every graph
with $\Omega(n^2)$ edges, taking all the nodes leads to a trivial
constant MVC approximation. Thus, our lower bound construction in \Cref{sec:lower} requires the construction of a ``fractal'',
self-recursive graph that is neither too symmetric nor too asymmetric,
and has a variety of node degrees! To the best of our knowledge, not
many graphs with these ``non-properties'' are known in computer
science, where symmetry and regularity are often the key to a
solution.

\textbf{Upper Bounds: }It is interesting to compare the lower and
upper bounds for the various problems. The MVC algorithm presented
in \Cref{sec:mvc-upper} achieves an $O(\Delta^{1/k})$
approximation in $k$ communication rounds, and hence,
the lower and upper bounds achieved in
\Cref{thm:mvc-lowerbound,thm:mvc-approxratio} are almost
\emph{tight}. In particular, any distributed algorithm requires at
least $\Omega(\log\Delta/\log\log\Delta)$-hop neighborhood
information in order to achieve a
constant or polylogarithmic approximation ratio to the MVC problem, respectively, which is
exactly what our algorithm achieves for polylogarithmic
  approximation ratios. It has recently been shown that
  even a $(2+\eps)$-approximation for the MVC problem can be computed
  in time $O(\log\Delta/\log\log\Delta)$ and thus our lower bound is
  also tight for constant approximation ratios \cite{baryehuda16}.

Our bounds are not equally tight when expressed as a function of
$n$, rather than $\Delta$. In particular, the remaining gap between
our upper and lower bounds can be as large as $\Theta(\sqrt{\log
n/\log\log n})$. The
additional square-root in the lower bounds when formulated as a
function of $n$ follows inevitably from the high-girth construction
of $G_k$: In order to derive a lower-bound graph as described in
\Cref{sec:graph,sec:construction}, there must be
many ``bad'' nodes that have the same view as a few neighboring
``good'' nodes. If each bad node has a degree of $\delta_{bad}$ (in
$G_k$, this degree is $\delta_{bad}\in\Theta(n^{1/k})$) and if we
want to have girth at least $k$, the graph must contain at least
$n\geq \delta_{bad}^k$ nodes. Taking all good nodes and applying
\Cref{alg:mvc_fmm} of \Cref{sec:mvc-upper} to the
set of bad nodes, we obtain an approximation ratio of $\alpha\in
O(\delta_{bad}^{1/k})$ in $k$ communication rounds. Combining this
with the bound on the number of nodes in the graph, it follows that
there is no hope for a better lower bound than $\Omega(n^{1/k^2})$
with this technique. From this it follows that if we want to improve
the lower bound (i.e., by getting rid of its square-root), we either
need an entirely different proof technique, or we must handle graphs
with low girth in which nodes do not see trees in their $k$-hop
neighborhood, which would necessitate arguing about views containing
cycles.

\textbf{Future Work: }We believe that the study of local computation
and local approximation is relevant far beyond distributed
computing, and there remain numerous directions for future research.
Clearly, it is interesting to study the locality of other network
coordination problems that appear to be polylog-local, including for
example the \emph{maximum domatic partition problem} \cite{feige03},
the \emph{maximum unique coverage problem}~\cite{demaine-soda06},
or various coloring problems~\cite{coloringmonograph}.

Beyond these specific open problems, the most intriguing \emph{distant
  goal} of this line of research is to divide distributed problems
into \emph{complexity classes} according to the problems' local
nature. The existence of \emph{locality-preserving reductions} and the
fact that several of the problems discussed in this paper exhibit
similar characteristics with regard to local
computability/approximability raises the hope for something like a
\emph{locality hierarchy} of combinatorial optimization problems. It
would be particularly interesting to establish ties between such a
distributed hierarchy of complexity classes and the classic complexity
classes originating in the Turing model of
computation~\cite{Lenzen2008What}. A first step in this direction has
recently been done in \cite{distributeddecision},
where complexity classes for distributed decision problems were
defined. Note that unlike in standard sequential models, in a
distributed setting, the complexity of decision problems is often not
related to the complexity of the corresponding search problems.

Besides classifying computational problems, studying local computation
may also help in gaining a more profound understanding of the relative
strengths of the underlying \emph{network graph models} themselves. It
was shown in
~\cite{Schneider2008Log-Star}, for example, that a MIS can be
computed in unit disk graphs (as well as generalizations thereof) in
time $O(\log^* \! n)$, which---in view of Linial's lower bound on
the ring---is asymptotically optimal. Hence, in terms of local
computability, the vast family of unit disk graphs are equally hard
as a simple ring network. On the other hand, our lower bounds prove
that general graphs are strictly harder, thus separating these
network topologies.

\section{Acknowledgements}
We thank the anonymous reviewers of the paper for various helpful
comments. We also thank Mika G\"o\"os for bringing the common lifts
construction by \cite{angluin81} to our attention (used to simplify
the construction in \Cref{sec:construction}). We are also grateful to
Bar-Yehuda, Censor-Hillel, and Schwartzman \cite{baryehuda16} for
pointing out an error in an earlier draft \cite{previousversion} of
this paper.

\clearpage
